\documentclass[a4paper,twocolumn,11pt,accepted=2025-02-12]{quantumarticle}
\pdfoutput=1
\usepackage[utf8]{inputenc}
\usepackage[english]{babel}
\usepackage[T1]{fontenc}
\usepackage{amsmath}
\usepackage{hyperref}
\usepackage[numbers,sort&compress]{natbib}
\usepackage{tikz}
\usepackage{lipsum}

\usepackage{amssymb,amsfonts}
\usepackage{algorithmic}
\usepackage{graphicx}
\usepackage{textcomp}
\usepackage{xcolor}
\newcommand{\norm}[1]{\left\lVert#1\right\rVert}
\newcommand{\abs}[1]{\left\lvert#1\right\rvert}
\newcommand{\comma}{,}
\newcommand{\QLayer}[3]{\gategroup[#1,steps=#2,style={dashed,rounded corners,fill=blue!20, inner xsep=2pt},background,
label style={label position=below,anchor=north,yshift=-0.2cm}]{{#3}}}

\usepackage{amsthm} 
\usepackage{amssymb}
\usepackage{pgfplots}
\usepgfplotslibrary{fillbetween}
\usepackage{float}
\usepackage{braket}
\pgfplotsset{compat=1.18}

\usepackage{tikz}
\usetikzlibrary{quantikz}
\usetikzlibrary{shapes.geometric, arrows}
\tikzstyle{startstop} = [rectangle, rounded corners, minimum width=3cm, minimum height=1cm,text centered, draw=black, fill=red!30]
\tikzstyle{io} = [trapezium, trapezium left angle=70, trapezium right angle=110, minimum width=3cm, minimum height=1cm, text centered, draw=black, fill=blue!30]
\tikzstyle{process} = [rectangle, minimum width=3cm, minimum height=1cm, text centered, draw=black, fill=orange!30]
\tikzstyle{decision} = [diamond, minimum width=3cm, minimum height=1cm, text centered, draw=black, fill=green!30]
\tikzstyle{blank_bounds} = [rectangle,rounded corners=0cm, minimum width=0cm, minimum height=0cm,text centered, draw=black, fill=white, line width=.6]
\tikzstyle{blank_rectangular} = [rectangle, minimum width=0cm, minimum height=0cm,text centered, draw=black, fill=white]
\tikzstyle{meter} = [rectangle,rounded corners=0cm, inner sep=10,  fill=white, minimum width=30,draw=black, line width=.6,
 path picture={\draw[black] ([shift={(.1,.3)}]path picture bounding box.south west) to[bend left=50] ([shift={(-.1,.3)}]path picture bounding box.south east);
 \draw[black,-latex] ([shift={(0,.1)}]path picture bounding box.south) -- ([shift={(.3,-.1)}]path picture bounding box.north);}]
\tikzstyle{arrow} = [thick,->,>=stealth]
\tikzstyle{arrow_} = [thick,-,>=stealth]
\newtheorem{definition}{Definition}
\newtheorem{thm}{Theorem}

\newtheorem{lemma}{Lemma}

\newtheorem{innercustomthm}{Theorem}
\newtheorem*{remark}{Remark}
\newenvironment{customthm}[1]{\renewcommand\theinnercustomthm{#1}\innercustomthm}
  {\endinnercustomthm}
\usepackage{amsfonts}
\DeclareMathOperator*{\argmin}{arg\,min}
\usepackage{hyperref}
\usepackage{url}
\usepackage{dsfont}
\usepackage{graphicx}% Include figure files

\usepackage{adjustbox}
\usepackage{dcolumn}% Align table columns on decimal point
\usepackage{bm}% bold math

\begin{document}

\title{Approximation and Generalization Capacities of Parametrized Quantum Circuits for Functions in Sobolev Spaces}

\author{Alberto Manzano}
\affiliation{Department of Mathematics and CITIC, Universidade da Coruña, Campus de Elviña s/n, A Coruña, Spain}
\email{alberto.manzano.herrero@udc.es}

\author{David Dechant}
\affiliation{$\langle aQa^L\rangle$ Applied Quantum Algorithms Leiden, The Netherlands}
\affiliation{Instituut-Lorentz, Universiteit Leiden, P.O. Box 9506, 2300 RA Leiden, The Netherlands} 

\author{Jordi Tura}
\affiliation{$\langle aQa^L\rangle$ Applied Quantum Algorithms Leiden, The Netherlands}
\affiliation{Instituut-Lorentz, Universiteit Leiden, P.O. Box 9506, 2300 RA Leiden, The Netherlands}

\author{Vedran Dunjko}
\affiliation{$\langle aQa^L\rangle$ Applied Quantum Algorithms Leiden, The Netherlands}
\affiliation{LIACS, Universiteit Leiden, P.O. Box 9512, 2300 RA Leiden, Netherlands}

\maketitle

\begin{abstract}
Parametrized quantum circuits (PQC) are quantum circuits which consist of both fixed and parametrized gates. 
In recent approaches to quantum machine learning (QML), PQCs are essentially ubiquitous and play the role analogous to classical neural networks. 
They are used to learn various types of data, with an underlying expectation that if the PQC is made sufficiently deep, and the data plentiful, the generalization error will vanish, and the model will capture the essential features of the distribution. 
While there exist results proving the approximability of square-integrable functions by PQCs under the $L^2$ distance, the approximation for other function spaces and under other distances has been less explored. 
In this work we show that PQCs can approximate the space of continuous functions, $p$-integrable functions and the $H^k$ Sobolev spaces under specific distances. 
Moreover, we develop generalization bounds that connect different function spaces and distances. 
These results provide a theoretical basis for different applications of PQCs, for example for solving differential equations. 
Furthermore, they provide us with new insight on the role of the data normalization in PQCs and of loss functions which better suit the specific needs of the users.
\end{abstract}

\section{Introduction}
Machine learning has gained significant attention in recent years for its practical applications and transformative impact in various fields.
As a consequence, there has been a rising interest in exploring the use of quantum circuits as machine learning models, capitalizing on the advancements in both fields to unlock new possibilities and potential breakthroughs.
Among the various possibilities for leveraging quantum circuits in machine learning, our particular focus lies on parametrized quantum circuits (PQC).
These quantum circuits consist of both fixed and adjustable (hence 'parametrized') gates.
When used for a learning task such as learning a function \cite{mitarai2018quantum}, a classical optimizer updates the parameters of the PQC in order to minimize a cost function depending on measurement results from this quantum circuit (see Figure \ref{fig: sketch of circuit and variational algorithm}).\\ \\

\begin{figure}
    \centering
    \begin{tikzpicture}[node distance=2.5cm, scale=1, every node/.style={transform shape}]
        \node (circuit) [startstop] {
        \begin{tikzpicture}[node distance=1.5cm]
            \node (initial) [blank_bounds,opacity=.0,text opacity=1.] {$\ket{0}^{\otimes n}$};
            \node (unitary) [blank_bounds, right of=initial] {$U(x,\bm{\theta})$};
            \node (measure) [meter,right of=unitary]{};
            \draw [arrow_] (initial) -- (unitary);
            \draw [arrow_] (unitary) -- (measure);
        \end{tikzpicture}
        };
        \node (expectation) [startstop, right of=circuit, xshift=2cm,fill=green!30] {$f_{\bm{\theta}}(x)$};
        \node (cost) [startstop, below of=expectation,fill=blue!30] {$D(f^*,f_{\bm{\theta}})$};
        \draw [arrow] (circuit) -- (expectation);
        \draw [arrow] (expectation) -- (cost);
        \draw [arrow] (cost) -- (0,-2.5) -- (circuit);
    \end{tikzpicture}

    \caption{Sketch of a hybrid variational algorithm. $U(x,\mathbb{\theta})$ represents a quantum circuit that takes $x$ as input and with variational parameters $\mathbb{\theta}$, $f_{\mathbb{\theta}}(x)$ is the expected value of some observable and $D(f^*,f_{\mathbb{\theta}})$ is the expected loss that we want to minimize. }
    \label{fig: sketch of circuit and variational algorithm}
\end{figure}
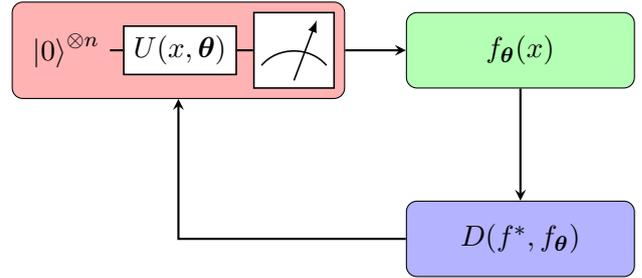
In this context, a growing line of research studies the expressivity of PQCs. 
More precisely, the capacity of PQCs to approximate any function belonging to a particular \textit{function space} defined in a prescribed \textit{domain} up to arbitrary precision with respect to a specific \textit{distance}. 
In \cite{schuld2021effect}, they showed that PQCs can be written as a generalized trigonometric series in the following way: 
\begin{equation}
\begin{aligned}\label{eq:circuit_fourier}
    f_{\bm{\theta}}(\bm{x}) &= \bra{0}U^{\dagger}(\bm{x};\bm{\theta})MU(\bm{x};\bm{\theta}) \ket{0} \\
    &=\sum_{\bm{\omega} \in \Omega} c_{\bm{\omega}}(\theta)e^{i\bm{\omega x}}\ .
\end{aligned} 
\end{equation}
 We would like to emphasize that although similar, the form of the PQC in above equation is more general than a Fourier series. This will become relevant for the results of this work.
 Using this formulation, it was further shown in \cite{schuld2021effect} that, if the PQC is chosen carefully, the increase of its depth and number of parameter can arbitrarily reduce the $L^2$ distance between the expected value $f_{\bm{\theta}}(\bm{x})$ of the PQC and any square-integrable function with the domain $[0,2\pi]^N$.
Throughout the paper we will refer to the PQC as the one approximating the functions to make the text more fluent, although technically it is the expectation value of the PQC that approximates the function.\\
This result had a significant impact on the motivation to study PQC-based QML, analogous to the impact that the famous Universality theorem for neural networks of Cybenko \cite{cybenko1989approximation} had on the domain of classical machine learning.
 Previously, different results on universality for PQC have been established. In \cite{du2020expressive}, the power of PQCs in expressing matrix product states and instantaneous quantum polynomial circuits was shown. Later, the universal approximation of PQCs was studied in regression problems, for single-qubit circuits with multiple layers \cite{perez2021one,yu2022power} and for both single- and multiple-qubit circuits \cite{goto2021universal,casas2023multidimensional}. 
However, as it turns out, there are numerous different notions of universality, and not all are useful for all applications. 
For instance, as will be discussed later, in the context of Physics-Inspired Neural Networks (PINN) the "vanilla" universality does not suffice.
This raises the question of whether PQCs can approximate functions belonging to other function spaces or in terms of other distances. \\ 

In this paper, we present two novel results. 
The first result of this paper is that PQCs can arbitrarily approximate the space of continuous functions, the space of $p$-integrable functions and the $H^k$ space, which is the set of functions whose derivatives up to order $k$ are square integrable.\\
Furthermore, we explain how these properties can be easily achieved in practice by a simple min-max feature rescaling (see \eqref{eq:min_max_feature_scaling}) of the input data.
In practice, this leads to an improved expressivity of PQCs, if the input data is normalized accordingly.\\ 

The second result of the paper are generalization bounds that connect distances with loss functions which are not built via the discretization of the integrals present in the definition of the distance. 
To make it more clear, we recall that in a machine learning problem one needs to choose an architecture, which defines the class of functions that can be approximated, and a target distance, which is intimately connected with the generalization error\footnote{In practice we may not explicitly think about the target distance, i.e. with respect to which distance we wish to approximate the "true" labeling function.
But this decision is implicitly made, once the loss is chosen.}. 
However, in general it is not possible to compute the target distance, as we would need to have available infinitely many data points.
Instead, one chooses a different distance function which can be computed from the available data: a loss function. 
This loss function is a different function than the target distance but it should be chosen in such a way that we call \textit{consistent} with the target distance, i.e., that the minimization of the expectation value of the loss function, the expected loss, yields the minimization of the target distance up to an error which asymptotically tends to zero \textit{when the number of samples and the expressivity} (here meant architecturally, as e.g. depth) of the PQC increases.
For example, the mean square error (as a loss function) is consistent with the $L^2$ error (as a target distance) but is inconsistent with the supremum distance. 
The usual generalization bounds connect target distances which are continuous with expected losses which are their discrete version. \\
The generalization bounds that we derive give a mapping \textit{across} different distances and loss functions, i.e., they relate distances with loss functions which are not built via the discretization of the integrals present in the definition of the distance. 
A particular loss function we shall define, denoted $\ell_{h^1}$, which consists of the sum of the mean square errors of the values of the functions and its derivative, is consistent with the supremum distance in one-dimensional problems.
In the described case, this allows us to reduce the supremum distance while choosing a loss function which is differentiable. \\ 

Our results apply in many settings.
For example, our first result has a direct consequence in that it allows one to approximate not most, but all function values with satisfying quality. 
For instance, the minimization of the ubiquitous $L^2$ distance may allow functions to dramatically differ from the target function in some regions where we have plenty of data points available, whereas the minimization of the supremum norm in Theorem \ref{thm:C_0_approximation} will force the PQC to converge for any given point in the domain of the target function. 
This is of high relevance in cases where we are interested in having a good approximation at any given point. 
For instance, when learning the shape of a probability distribution from samples, a good fit in the bulk of the distribution but not in its tails can lead to significant underestimation or overestimation of the probability of extreme events. 
In real-world applications, this could have severe consequences in risk assessment applications, where accurate estimation of tail probabilities is essential for developing appropriate contingency measures against rare but significant events, such as the COVID-19 pandemic or the 2008 economic crisis.
Our second result has direct applications, e.g., in settings where we have access to data of the function and its derivatives.
One case where this is standard is in settings involving solving differential equations.
For example in physics-informed neural networks (PINN) problems \cite{Raissi2017PINN} and differential machine learning (DML) \cite{huge2020differential}, both function values and derivatives are accessible and in fact critical. \\ 
 
This paper is organized as follows: in Section \ref{sec:data_normalization} we explain the new results on the expressivity of PQCs. 
In Section \ref{sec:differential_machine_learning} we discuss the proposed generalization bounds. 
Then, in Section \ref{sec:numerical_experiments} we illustrate the theoretical result of Sections \ref{sec:data_normalization} and \ref{sec:differential_machine_learning} by means of some numerical experiments. 
Lastly, in Section \ref{sec:conclusions} we wrap up with the conclusions. \\ 

During the final stages of our work, we became aware of the paper \cite{gonon2023universal} which overlaps in some parts with our own results in Section \ref{sec:data_normalization}. However, the results presented here were developed independently and follow a different line of reasoning.

\section{PQCs and universal approximation}
\label{sec:data_normalization}
In this section, we will review the established result on universality in \cite{schuld2021effect} and then present our new universality results in Theorems \ref{thm:L_p_approximation}, \ref{thm:C_0_approximation} and \ref{thm:H^k_approximation}.\\ \\
Schuld et al. showed in \cite{schuld2021effect}, how a quantum machine learning model of the form $f_{\theta}(x)=\bra{0}U^{\dagger}(x;\bm{\theta})MU(x;\bm{\theta}) \ket{0}$ can be written as a univariate generalized trigonometric series:
\begin{align}
    \bra{0}U^{\dagger}(x;\bm{\theta})MU(x;\bm{\theta}) \ket{0} &= f_{m}(x;\bm{\theta}) \\
    &=\sum_{\bm{\omega} \in \Omega} c_{\bm{\omega}}(\theta)e^{i\bm{\omega x}},
\end{align}
where $M$ is an observable, $U(x;\bm{\theta})$ is a quantum circuit modeled as a unitary that depends on input $x$ and the variational parameters $\bm{\theta}=\left(\theta_0,\theta_1,...,\theta_T \right)$.
In the above, $\bm{\omega}\in\Omega$ denotes the set of available frequencies which always contain $0$. 
The quantum circuit consists of $L$ layers each consisting of a trainable circuit block $W_i(\bm{\theta}), i\in\{1,...,L+1\}$ and a data encoding block $S(x)$ as shown in Figure \ref{fig:Schuld_circuit}.
The data encoding blocks determine which frequencies $\bm{\omega}$ are accessible in the sum and are implemented as Pauli rotations.
The blocks $W(\bm{\theta})$ can be built from single-qubit rotation gates and CNOT gates and they determine the coefficients $c_{\bm{\omega}}(\theta)$ of the sum. It is possible to both implement this model with $L>1$ layers, such as data re-uploading PQC \cite{gil2020input,perez2020data}, where the encoding is repeated on the same subsystems in sequence, or with parallel encodings \cite{rebentrost2014quantum} and $L=1$, where the encoding is repeated on several different subsystems.\\ \\

\begin{figure*}
   \centering
    \begin{adjustbox}{width=0.8\textwidth}
    \begin{quantikz}[scale=0.6, every node/.style={transform shape},column sep=0.2cm, row sep={1cm,between origins}]
        %%%%%%%%%%%%%%%%%%%%%%%%%%%%%%%%%%%%%%%%%%%%%%%%%%%%%%%%%%%%%%%%%%%%%%%%%%%%%%%%%%%%%%%%%%%%%%%%%%%%%%%%%%%%%%%%%%
        \lstick{$\ket{0}$} 
        & \phantomgate{}    & \gate[4,nwires=3]{S(x)} \QLayer{4}{2}{Layer 1} & \gate[4,nwires=3]{W_1(\bm{\theta})}    
        & \phantomgate{}    & \gate[4,nwires=3]{S(x)} \QLayer{4}{2}{Layer 2} & \gate[4,nwires=3]{W_2(\bm{\theta})}     
        & \phantomgate{}    & \gate[4,nwires=3]{S(x)} \QLayer{4}{2}{Layer L} & \gate[4,nwires=3]{W_L(\bm{\theta})}    
        & \phantomgate{}    & \meter{} \\
        %%%%%%%%%%%%%%%%%%%%%%%%%%%%%%%%%%%%%%%%%%%%%%%%%%%%%%%%%%%%%%%%%%%%%%%%%%%%%%%%%%%%%%%%%%%%%%%%%%%%%%%%%%%%%%%%%%
        \lstick{$\ket{0}$} 
        & \phantomgate{}    &                                                    & 
        & \phantomgate{}    &                                                    & 
        & \phantomgate{}    &                                                    &
        & \phantomgate{}    & \meter{}\\
        %%%%%%%%%%%%%%%%%%%%%%%%%%%%%%%%%%%%%%%%%%%%%%%%%%%%%%%%%%%%%%%%%%%%%%%%%%%%%%%%%%%%%%%%%%%%%%%%%%%%%%%%%%%%%%%%%%
        \vdots
        &                   &                                                    & 
        &                   &                                                    & 
        &\qquad\hdots\qquad &                                                    &
        &                   & \vdots\\
        %%%%%%%%%%%%%%%%%%%%%%%%%%%%%%%%%%%%%%%%%%%%%%%%%%%%%%%%%%%%%%%%%%%%%%%%%%%%%%%%%%%%%%%%%%%%%%%%%%%%%%%%%%%%%%%%%%
        \lstick{$\ket{0}$} 
        & \phantomgate{}    &                                                    &
        & \phantomgate{}    &                                                    &
        & \phantomgate{}    &                                                    & 
        & \phantomgate{}    &\meter{}
        %%%%%%%%%%%%%%%%%%%%%%%%%%%%%%%%%%%%%%%%%%%%%%%%%%%%%%%%%%%%%%%%%%%%%%%%%%%%%%%%%%%%%%%%%%%%%%%%%%%%%%%%%%%%%%%%%%
    \end{quantikz}
    \end{adjustbox}
    \label{fig:circuit_scheme}
   \caption{parametrized quantum circuit that can be written as a generalized trigonometric series as in \eqref{eq:circuit_fourier}. It consists of $L$ layers, each layer is composed by a trainable circuit block $W_i(\bm{\theta}), i\in\{1,...,L+1\}$ and a data encoding block $S(x)$. The data encoding blocks $S(x)$ are identical for all layers, they determine which frequencies $\bm{\omega}$ are accessible and are implemented as Pauli rotations. The blocks $W_i(\bm{\theta})$ can be built from local rotation gates and $CNOT$ gates. They determine the coefficients $c_{\bm{\omega}}(\theta)$.}
   \label{fig:Schuld_circuit}
\end{figure*}
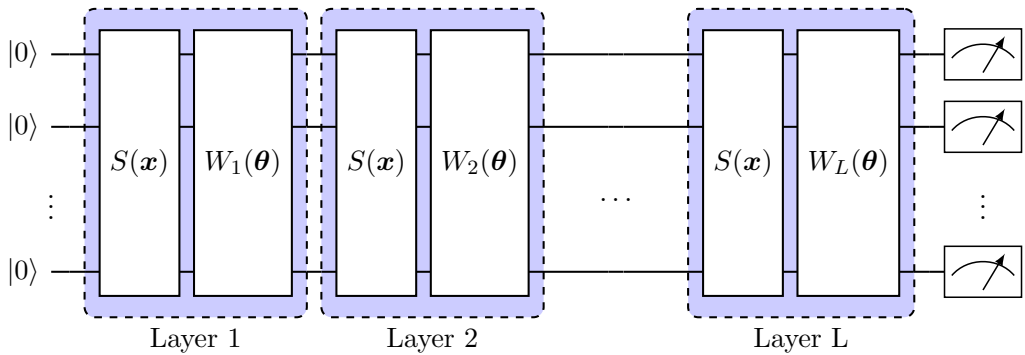

For the needs of our discussion, we will briefly describe a more specific set-up under which the authors of \cite{schuld2021effect} proved a universality theorem of these quantum models for the multivariate case with inputs $\bm{x} = \left(x_0,x_1,...,x_N\right)$.\\
Let us construct a model of the form in \eqref{eq:circuit_fourier},
with the measurement $M$ and a quantum circuit of one layer, $L=1$:
\begin{align}
    f_{\bm{\theta}}&=\bra{0}U^{\dagger}(\bm{\theta},\bm{x})MU(\bm{\theta},\bm{x})\ket{0}\ ,\text{ with }\\
    U(\bm{\theta},\bm{x})&=W^{(2)}(\bm{\theta}^{(2)})S(\bm{x})W^{(1)}(\bm{\theta}^{(1)})\ ,
\end{align}
where $\bm{\theta}^{(1)}$ and $\bm{\theta}^{(2)}$ are those parameters in $\bm{\theta}$ that affect $W^{(1)}$ and $W^{(2)}$, respectively. Let us further make the following two assumptions:
Firstly, we assume that the data-encoding blocks $S(\bm{x})$ are written in the following way:
\begin{align}
    S(\bm{x})&=e^{-x_1H}\otimes \cdots\otimes e^{-x_NH}\\
    &=:S_H(\bm{x})\ ,
\end{align}
where $H$ is a Hamiltonian that we specify later.
Secondly, we assume that the trainable circuit blocks $W^{(1)}(\bm{\theta}^{(1)})$ and $W^{(2)}(\bm{\theta}^{(2)})$ are able to represent arbitrary global unitaries.
In practice, this may require exponential circuit depth.
With this assumption, we drop the dependence on $\bm{\theta}$ and reformulate the assumption as being able to prepare an arbitrary initial state $\ket{\Gamma}:=W^{(1)}(\bm{\theta}^{(1)})\ket{0}$ and by absorbing $W^{(2)}(\bm{\theta}^{(2)})$ into the measurement $M$.
We can then write the above quantum model as:
\begin{align}
   f(\bm{x})=\bra{\Gamma}S_H^{\dagger}(\bm{x})MS_H(\bm{x})\ket{\Gamma} \ .
\end{align}
Let us further present the notion of a universal Hamiltonian family, as defined in \cite{schuld2021effect}:
\newpage
\begin{definition}
\label{definition: universal hamiltonian family} Let $\{H_m|m\in\mathbb{N}\}$ be a Hamiltonian family where $H_m$ acts on $m$ subsystems of dimension $d$. \\
Such a Hamiltonian family gives rise to a family of models $\{f_m\}$ in the following way:
\begin{align}\label{eq:family_of_models_dep_on_Hamiltonian_family}
   f_m(\bm{x)}=\bra{\Gamma}S_{H_m}^{\dagger}(\bm{x})MS_{H_m}(\bm{x})\ket{\Gamma} \ .
\end{align}
Further, we call the set
\begin{align}
    \Omega_{H_m}:=\{\mathbf{\lambda}_j-\mathbf{\lambda}_k|j,k\in\{1,...,d^m\}\}\ 
\end{align}
where $\{\lambda_1,...,\lambda_{d^m}\}$ are the eigenvalues of $H_m$, the frequency spectrum of $H_m$.\\
\end{definition}
\begin{remark}
We call a Hamiltonian family $\{H_m\}$
a universal Hamiltonian family, if for all $K\in\mathbb{N}$, there exists an $m\in\mathbb{N}$, such that:
\begin{align}
\mathbb{Z}_K=\{-K,...,0,...,K\}\subseteq\Omega_{H_m}\ ,
\end{align}
hence if the frequency spectrum of $\{H_m\}$ asymptotically contains any integer frequency.
\end{remark}
As shown in \cite{schuld2021effect}, a simple example of a universal Hamiltonian family is one which consists of tensor products of single-qubit Pauli gates:
\begin{align}
    H_m=\sum_{i=1}^m\sigma_{q}^{(i)}\ ,
\end{align}
with $\sigma_q^{(i)}$, $q\in\{X,Y,Z\}$ and $d=2$. The scaling of the frequency spectrum for this example goes as $K=m$.\\
With these definitions, we can give the following theorem:
\begin{thm}[Convergence in $L^2$]\label{thm:Schuld_universality} \cite{schuld2021effect}
    Let $\{H_m\}$ be a universal Hamiltonian family, and $\{f_m\}$ the associated quantum model family, defined
via \eqref{eq:family_of_models_dep_on_Hamiltonian_family}. For all functions $f^*\in L^2\left([0, 2\pi]^N\right)$, and for all $\epsilon >0$, there exists some $m'\in \mathbb{N}$, some state $\ket{\Gamma} \in \mathbb{C}^{m'}$
and some observable $M$ such that
\begin{equation}\label{eq:schuld_l2_approximation}
   \norm{f_{m'}-f^*}_{L^2}<\epsilon.
\end{equation}
\end{thm}
Here, we clearly see that there are two conditions on the target function $f^*$ that must be fulfilled in order for the theorem to work properly.
The first condition is that $f^*$ belongs to $L^2$.
This is not surprising, we need to assume certain regularity on the target function to make the theorem work.
The second condition is that the target function $f^*$ needs to be restricted to the domain $[0,2\pi]^N$.
However, as suggested in the original paper \cite{schuld2021effect}, if the function $f^*$ does not belong to this domain, we can easily map $[a,b]^N$ to the required domain $[0,2\pi]^N$ (or $[-\pi,\pi]^N$ equivalently).\\ 

We would like to highlight the fact that the distance we use to bring the approximator closer to the target function is the $L^2$ distance.
Note that convergence in the $L^2$ sense does not imply other modes of convergence.
For example, this does not give us information about the general case of $L^p$-distances, with $1\leq p <\infty$.
We explicitly address this more general case in the following theorem:
\begin{thm}[Convergence in $L^p$]\label{thm:L_p_approximation}
    Let $\{H_m\}$ be a universal Hamiltonian family, and $\{f_m\}$ the associated quantum model family, defined
via \eqref{eq:circuit_fourier}. For all functions $f^*\in L^p\left([0, 2\pi]^N\right)$ where $1\leq p<\infty$, and for all $\epsilon >0$, there exists some $m'\in \mathbb{N}$, some state $\ket{\Gamma} \in \mathbb{C}^{m'}$,
and some observable $M$ such that:
\begin{equation}\label{eq:L_p_approximation}
   \norm{f_{m'}-f^*}_{L^p}<\epsilon.
\end{equation}
\end{thm}
The proof of Theorem \ref{thm:L_p_approximation} is given in Appendix \ref{appendix:fourier_series_uniform_approximation}.

Let us emphasize the difference between Theorems \ref{thm:Schuld_universality} and \ref{thm:L_p_approximation}:
The target function can belong to any $L^p$ space with $1\leq p<\infty$ in contrast to the previous requirement of being square-integrable ($L^2$). 
This is essentially achieved by the fact that PQCs are not only able to represent Fourier series as it is discussed in \cite{schuld2021effect} but they are also able to represent more general trigonometric series. 
This allow us to identify the expectation value of the quantum circuit with the Cèsaro summation of the partial Fourier series of $f^*$ and leverage the power of Fejér-like theorems \cite{fejer1903untersuchungen}. 
See Appendix \ref{appendix:fourier_series_uniform_approximation} for more details.\\
Nevertheless, the ability to approximate functions in $L^p$ does not prevent us from having arbitrarily big errors in certain points.
Intimately related to this problem is the so-called Gibbs phenomenon \cite{gibbs}.
Namely, the approximation of a continuous, but non-periodic function by a Fourier series is increasingly better in the interior of the domain but increasingly poorer on its boundaries. 
That leads to the fundamental question if we can approximate $f^*$ in a stronger sense, so that we ensure that the target function $f^*$ is well approximated in any given point. We answer this question in the next theorem.

\begin{thm}[Convergence in $C^0$]\label{thm:C_0_approximation}
    Let $\{H_m\}$ be a universal Hamiltonian family, and $\{f_m\}$ the associated quantum model family, defined
via \eqref{eq:circuit_fourier}. For all functions $f^*\in C^0\left(U\right)$ where $U$ is compactly contained in the closed cube $[0,2\pi]^N$, and for all $\epsilon >0$, there exists some $m'\in \mathbb{N}$, some state $\ket{\Gamma} \in \mathbb{C}^{m'}$,
and some observable $M$ such that $f_{m'}$ converges uniformly to $f^*$:
\begin{equation}\label{eq:C_0_approximation}
   \norm{f_{m'}-f^*}_{C^0}<\epsilon,
\end{equation}
with \footnote{Since $f^*$ is defined on a compact domain $U$, the supremum is equivalent to the maximum in this case.}
\begin{align}
    \norm{f_{m'}-f^*}_{C^0}:=\sup_{\mathbf{x}\in [0,2\pi]^N}\|f_{m'}(\mathbf{x})-f^*(\mathbf{x})\| \ .
\end{align}
\end{thm}
The proof of Theorem \ref{thm:C_0_approximation} can be found in Appendix \ref{appendix:fourier_series_uniform_approximation}.\\
A set $U\subset\mathbb{R}^N$ is compactly contained in another set $V\subset\mathbb{R}^N$, if the closure of $U$ is compact and contained in the interior of $V$.\\
Simply stated, this theorem means that $f_{m'}$ converges uniformly to $f^*$.
In other words, if we select a given target error $\epsilon$ we are always able to find a finite PQC such that the error on any point is smaller than the prescribed $\epsilon$. 
Let us emphasize again the differences between Theorems \ref{thm:Schuld_universality} and \ref{thm:C_0_approximation}. 
The first difference is that the function $f^*$ has to be defined in a domain $U$ which is compactly contained in $\left[0,2\pi\right]^N$. 
A simple example of $U$ is the interval $\left(\left[-\frac{\pi}{2}, \frac{\pi}{2}\right]^N\right)$ (or $\left([0, \pi]^N\right)$, equivalently).
By restricting ourselves to half of the original space we can always find a $C^0$ extension of the function $f^*$ in $\mathbb{T}^N$. 
The second difference is that the target function now belongs to the class of continuous functions in contrast to the previous requirement of being square-integrable ($L^2$). \\ 

A last result that we will show in this regard is about the approximation of the function and its derivatives by the parametrized quantum circuit.
This might seem as a purely synthetic question but it has many implications in practice.
When we approximate a target function, in many occasions we not only want to recover its value but also its dynamics.
This is particularly relevant for problems in physics, where we typically have a differential equation which describes the behavior of the system.
As we will see in the following theorem, the universality results translate to functions defined in the Sobolev space $H^k$ as well:
\begin{definition} The Sobolev space $H^k(\Omega)$ is defined as the space of square integrable functions on a domain $\Omega\subseteq\mathbb{R}^N$ which derivatives up to order $k$ are square integrable as well:
\begin{align} 
f^{\alpha}=D^{\alpha}f\ , \text{ and }
\norm{f^{(\alpha)}}_2<\infty\ ,
\end{align}
for all $0\leq\abs{\alpha}\leq k$ and $D^{\alpha}:=\frac{\partial^{\abs{\alpha}}}{\partial x_1^{\alpha_1}...\partial x_N^{\alpha_N}}$.\\
The Sobolev norm $\norm{\cdot}_{H^k}$ is defined as
\begin{align}
    \norm{f}_{H^k}:=\left(\sum_{\abs{\alpha}\leq k}\int_{\Omega}\abs{D^{\alpha}f}^2\right)^{1/2}\ .
\end{align}
\end{definition}

\begin{thm}[Convergence in $H^k$]\label{thm:H^k_approximation}
    Let $\{H_m\}$ be a universal Hamiltonian family, and $\{f_m\}$ the associated quantum model family, defined
via \eqref{eq:circuit_fourier}. For all functions $f^*\in H^{k}\left(U\right)$ where $U$ is compactly contained in the closed cube $[0,2\pi]^N$, and for all $\epsilon >0$, there exists some $m'\in \mathbb{N}$, some state $\ket{\Gamma} \in \mathbb{C}^{m'}$,
and some observable $M$ such that $f_{m'}$ converges to $f^*$ with respect to the $H^k-$distance:
\begin{equation}\label{eq:H^k_approximation}
   \norm{f_{m'}-f^*}_{H^k}<\epsilon.
\end{equation}
\end{thm}
The proof of Theorem \ref{thm:H^k_approximation} is given in Appendix \ref{appendix:fourier_series_uniform_approximation}.\\
As in Theorems \ref{thm:C_0_approximation} and \ref{thm:H^k_approximation}, we require that the target function is defined on a compactly contained subset of $[0,2\pi]^N$, we propose to perform a min-max feature scaling of the input data:
\begin{equation}
\begin{aligned}\label{eq:min_max_feature_scaling}
    \bm{x} = (x_1,...,x_n)
    \longrightarrow \ 
     \bm{\tilde{x}}= (\tilde{x}_1,...,\tilde{x}_n)\ ,
\end{aligned}
\end{equation}
where $\bm{x}\in [a,b]^N$, $\bm{\tilde{x}}\in \left[-\dfrac{\pi}{2},\dfrac{\pi}{2}\right]^N$, and 
\begin{align}
    \bm{\tilde{x}}= \left(-\dfrac{\pi}{2}+\pi\dfrac{x_1-a}{b-a},...,-\dfrac{\pi}{2}+\pi\dfrac{x_n-a}{b-a}\right)\ .
\end{align}
This simple recipe allows the PQC to approximate a much wider set of function spaces as shown throughout this section.
This normalization strategy works very well in practice as can be seen in Section \ref{sec:numerical_experiments}.
However, we would like to emphasize that this particular normalization is not the only choice.
The classical strategy in machine learning of normalizing the input data to lie in the $[-1,1]^N$ domain is also completely valid. \\
Throughout this section, we have discussed the expressive power of PQCs, but when we do machine learning, we have more ingredients that we need to take into account.
In the next section we will discuss the role that the loss function plays in accordance with the type of approximation that our PQC can get.

\section{Connections between different generalization bounds}\label{sec:differential_machine_learning}
As we have seen in the previous section, the notion of approximation depends on a prescribed distance.
This distance is not given by the problem itself, but instead chosen by the user, this is why we refer to it as target distance.
In general, it is however not possible to compute the target distance, which for example is the case for the $L^p$ and $H^k$ distances. 
This is why one needs to choose a distance function which can be computed from data, a loss function.
It has to be chosen in such a way that it is consistent with the target function.
To discuss the topic in more depth, let us formally introduce the continuous regression problem, which is the problem that we are most interested in.\\ 

In general, we can describe the continuous regression problem in the following way:
assume that there is some target function $f^*\in \mathcal{F}\subseteq H^k$ mapping inputs $x\in \mathcal{X}$ to target labels $y\in \mathcal{Y}$.
Moreover, assume that the points in $\mathcal{X}$ are sampled according to a bounded\footnote{It is possible to have more general density functions. However, we restrict ourselves with this one since it simplifies the analysis.} density function $p$.
Our goal is to find the best approximation $f\in \mathcal{M}\subseteq H^k$ of the target function $f^*$. 
\\ The notion of what is understood as a ``good'' approximation as clarified, allows for some freedom.
For this reason, one has to make a choice by specifying a functional $D : H^k\; \times \;  H^k \longrightarrow \mathbb{R}^+\cup \{0\}$ which defines a distance between the elements of $\mathcal{F}$ and $\mathcal{M}$.
The problem can then be stated as:
\begin{equation}\label{eq:minimization_problem}
    f = \argmin_{\hat{f}\in \mathcal{M}}D(f^*,\hat{f}).
\end{equation}

The most common distance in the literature for continuous regression problems is the one induced by the $L^2(\mathcal{X},P)$ norm:
\begin{equation}\label{eq:l2_problem}
\begin{aligned}
    D_{L^2}\left(f^*,f\right) &= ||f^*-f||_{L^2}\\ &= \left(\int_{\mathcal{X}}(f^*(\bm{x})-f(\bm{x}))^2dP\right)^{\frac{1}{2}}.
\end{aligned}
\end{equation}
However, in regression we do not typically have access to the full information (i.e., we cannot compute the integral).
It is for this reason that instead we work with the empirical risk minimization problem, which uses the discrete version $l^2$ of the $L^2$ distance as a loss function.
The difference with the previous setup is that, for the empirical risk minimization problem, we are given a finite training set $S$ of $I$ inputs sampled from the same probability density $p$, together with their target labels $\{(x_1, y_1),...,(x_I, y_I)\}$ with $(x, y)\in \mathcal{X}\times \mathcal{Y}$, according to the target function $f^*:\mathcal{X}\to \mathcal{Y}$, $f\in\mathcal{F}$.
Now, instead of minimizing a continuous functional, we will minimize a discrete one.
We call
\begin{align}
    D_{\ell}(f^*,f)=\frac{1}{I}\sum_{i = 0}^{I-1}  \ell \left( f^*(\bm{x}^i),f(\bm{x}^i)\right)
\end{align}
the expected loss according to a loss function $\ell:\mathcal{Y}\times \mathcal{Y}\to \mathbb{R}$.
Similarly to the continuous case, we are concerned with the expected loss of the $l^2$ distance, which is defined as:
\begin{align}\label{eq:empirical_risk_standard_loss}
    D_{l^2}(f^*,f) := \left(\frac{1}{I}\sum_{i = 0}^{I-1}  \left(f^*(\bm{x}^i)-f(\bm{x}^i)\right)^2\right)^{\frac{1}{2}},
\end{align}
with $\bm{x}^i$ denoting the i-th input.\\ 

Although we are solving the minimization problem associated with the expected loss defined in \eqref{eq:empirical_risk_standard_loss}, in general we are interested in the generalization performance, i.e., the distance in terms of \eqref{eq:l2_problem}.
Using generalization bounds \cite{mohri2018foundations} we can relate the performance in terms of the distance given by \eqref{eq:empirical_risk_standard_loss} with the distance given by \eqref{eq:l2_problem}.
However, these classical results in machine learning do in general not relate the $l^2$ distance with other distances, like the $C^0$ distance.
In other words, even a solution which, as the model and the number of points grow larger asymptotically makes the $D_{L^2}$ go to zero, does not necessarily make the $D_{C^0}$ distance vanish, which is defined as:
\begin{align}\label{eq:empirical_risk_uniform_loss}
    D_{C^0}(f^*,f) := \sup_{\mathbf{x}\in \mathcal{X}}\abs{f^*(\mathbf{x})-f(\mathbf{x})}.
\end{align}
In such cases, we could find points where there is an arbitrarily large discrepancy between the solution and the target function.\\
One possible solution would be to use a different distance than $D_{l^2}$. For example one could try the discrete form of the $D_{C^0}$ distance:
\begin{equation}\label{eq:max_distance}
    D_{l^\infty} = \max_{i\in \{0,...,I-1\}}\abs{f^*(\mathbf{x}^i)-f(\mathbf{x}^i)},
\end{equation}
but this distance is not differentiable, making the optimization process much harder.\\ 

Thus, we identify two desirable features for a distance in order to be able to approximate with the $C^0$ distance.
The first requirement is that the solution of the minimization problem that it defines, tends uniformly to the target function $f^*$ as we increase the number of given points $I$ and we increase the size of our PQC. 
The second one is that it has to be differentiable in order to make minimization easier.\\
The solution that we propose here is to use a distance motivated by discretizing the Sobolev distance $H^k$ on a fixed finite training set $\Big\{\left(x_1, f^*(\bm{x}^1),\{D^{\alpha}f^*(\bm{x}^1)\}_{|\alpha|\leq k}\right),...,\left(x_I, f^*(\bm{x}^I),\right.$\\ $\left.\{D^{\alpha}f^*(\bm{x}^I)\}_{|\alpha|\leq k}\right)\Big\}$, $(x, y)\in \mathcal{X}\times \mathcal{Y}$, according to the target function $f^*:\mathcal{X}\to \mathcal{Y}$, $f\in\mathcal{F}$. The sets $\{D^{\alpha}f(\bm{x})\}_{|\alpha|\leq k}$ and $\{D^{\alpha}f^*(\bm{x})\}_{|\alpha|\leq k}$ consists of the $M(N,k):=\sum_{\alpha=1}^k\binom{\alpha +N-1}{N-1}$ different partial derivatives up to order $k$ evaluated at point $\bm{x}$. We write $N$ for the number of input dimensions. Note that for being able to apply this distance, one needs to have access to training data containing the required partial derivatives additionally to the function values.\\ 
We show the expected loss of the discretized version of $H^1$ and $H^k$, respectively, in the following two equations:

\begin{widetext}
\begin{align}
\label{eq:h1_risk_problem}
 D_{h^1}(f^*,f)&:=\left[\frac{1}{I}\sum_{i = 0}^{I-1}\left(f^*(\bm{x}^i)-f(\bm{x}^i)\right)^2+ \sum_{j = 0}^{N-1}\sum_{i = 0}^{I-1}\frac{1}{I}\left(\dfrac{\partial f^*}{\partial x_j}(\bm{x}^i)-\dfrac{\partial f}{\partial x_j}(\bm{x}^i)\right)^2\right]^{\frac{1}{2}},\\
   \label{eq:hk_risk_problem}
    D_{h^k}(f^*,f)&:=\left[\frac{1}{I}\sum_{i = 0}^{I-1}\left(f^*(\bm{x}^i)-f(\bm{x}^i)\right)^2+ \sum_{\abs{\alpha}\leq k}\sum_{i = 0}^{I-1}\frac{1}{I}\left(D^{\alpha}f^*(\bm{x}^i)-D^{\alpha}f(\bm{x}^i)\right)^2\right]^{\frac{1}{2}}.
\end{align}
\end{widetext}

 The expected loss as given in \eqref{eq:h1_risk_problem} was first introduced in \cite{huge2020differential} and gives rise to a new subfield of machine learning known in the literature as differential machine learning (DML). 
 Its generalization, the discretization of the distance $H^k$, is given in \eqref{eq:hk_risk_problem}, and can be applied when the required higher-dimensional derivatives are available as well.
 The derivatives $\frac{\partial^{(p)} f^*}{\partial x_j^{(p)}}$ and $\frac{\partial^{(p)} f}{\partial x_j^{(p)}}$ are the $p$-th order derivative functions in direction $x_j$ of $f^*$ and $f$, respectively. The corresponding loss function is thus defined as
\begin{align}
    \label{eq:lossfunctionasmap}
    & \ell_{h^k}:\mathbb{R}^{M(N,k)+1}\times\mathbb{R}^{M(N,k)+1}\to \mathbb{R}, \\
        &\left(f(\bm{x}),\{D^{\alpha}f(\bm{x})\}_{|\alpha|\leq k},f^*(\bm{x}),\{D^{\alpha}f^*(\bm{x}):\}_{|\alpha|\leq k}\right)\mapsto\nonumber \\
        & \ell_{h^k}(f(\bm{x}),f^*(\bm{x}))=\left(f^*(\bm{x})-f(\bm{x}))\right)^2\\
        &+\sum_{\abs{\alpha}\leq k}\left(D^{\alpha}f^*(\bm{x})-D^{\alpha}f(\bm{x}))\right)^2\nonumber\ .
\end{align} 
With classical neural networks, DML has proven to yield better generalization results in terms of the $D_{l^2}$ distance than the solution of the $D_{l^2}$ itself.
This means that, if we take the solutions $f_{h^1}$ and $f_{l^2}$ of the minimization problems defined by Equations \eqref{eq:h1_risk_problem} with the same number of labels and \eqref{eq:empirical_risk_standard_loss} respectively and evaluate their performance in terms of the $D_{L^2}$, in practice $f_{h^1}$ performs better than $f_{l^2}$:
\begin{align}
    D_{L^2}\left(f^*,f_{h^1}\right)\leq D_{L^2}\left(f^*,f_{l^2}\right).
\end{align}
However, to the best of our knowledge there is no theoretical explanation in the literature on why this happens or under which condition we might expect this behavior. In the following theorems we present generalization bounds that shed some lights onto it.\\
Before stating them, we will define two function families to which the generalization bounds apply:
\begin{definition}\cite{Caro2021encodingdependent}
\label{def:functionfamilies}
     By $\mathcal{F}_{\Omega}^{B}$, we denote the function family defined as
\begin{align*}
    \mathcal{F}_{\Omega}^{B}=\Big\{&[0,2\pi]^N\ni \bm{x}\mapsto f(\bm{x})=\sum_{\bm{\omega}\in\Omega}c_{\bm{\omega}}\exp(-i\bm{\omega}\bm{x}):\\
    &\{c_{\bm{\omega}}\}_{\bm{\omega}\in\Omega} \text{ s.t. } \|f\|_{\infty}\leq B\text{ and }|\Omega|<\infty\Big\}\ .
\end{align*}
By $\mathcal{H}_{\Omega}^{\Tilde{B}}$, we denote the function family defined as
\begin{align*}
    \mathcal{H}_{\Omega}^{\Tilde{B}}= \Big\{&[0,2\pi]^N\ni \bm{x}\mapsto \frac{a_0}{2}\\
    &+\sum_{\bm{\omega}\in\Omega_+}(a_{\bm{\omega}}\cos(\bm{\omega}\bm{x})+b_{\bm{\omega}}\sin(\bm{\omega}\bm{x})):\\  
    &\sqrt{a_0^2+\sum_{\bm{\omega}\in\Omega_+}a_{\bm{\omega}}^2+ b_{\bm{\omega}}^2}\leq \tilde{B}\text{ and }|\Omega_+|<\infty\Big\}\ ,
\end{align*}
where the frequency set $\Omega$ is divided into the disjoint parts $\Omega=\Omega_+\cup \Omega_-\cup \{0\}$, where $\Omega_+\cap \Omega_-=\emptyset$ and such that for every $\bm{\omega}\in\Omega_+$, it holds that $-\bm{\omega}\in\Omega_-$.
\end{definition}
According to \cite{Caro2021encodingdependent}, both of these function families can be modeled by the quantum model given in Equation \eqref{eq:family_of_models_dep_on_Hamiltonian_family}. As can be seen by this equation, the bounds $B$ and $\Tilde{B}$
depend on the chosen circuit and observable, and they determine the scaling in the following generalization bounds.
Note as well that the truncated Fourier series as defined in $\mathcal{F}_{\Omega}^{B}$ and $\mathcal{H}_{\Omega}^{\Tilde{B}}$ are differentiable, and their derivatives form truncated Fourier series as well. 
If one chooses the frequency set $\Omega'$ and the bounds $B'$ and $\Tilde{B}'$ large enough, for a given function family $\mathcal{F}$, both the functions and their derivatives are contained in $\mathcal{F}_{\Omega'}^{B'}$ and $\mathcal{H}_{\Omega'}^{\Tilde{B}'}$.

\begin{thm}[Generalization bound for $H^k$]
\label{thm:sobolev_loss_Hk}

Let $f^*\in \mathcal{F}\subseteq H^k([0,2\pi]^N)$ be a target function, and let there be a $B>0$ and a $\Tilde{B}>0$, such that $\mathcal{F}_{\Omega}^{B}\subseteq \mathcal{H}_{\Omega}^{\Tilde{B}}$ is a suitable model family.
Let us further assume that $\ell_{h^k}(f_1(\bm{x}),f_2(\bm{x}))\leq c$ for all $\mathbf{x}\in [0,2\pi]^N$, and for all $f_1, f_2 \in\mathcal{F}_{\Omega}^{B}$ or $\mathcal{F}$. 
For any $\delta\in(0,1)$ and the empirical risk $D_{h^k}\left(f^*,f\right)$ trained on an i.i.d. training data $S$ with size $I$ and containing data of $\xi$ partial derivatives, the following holds for all functions $f\in \mathcal{F}_{\Omega}^{B}$ with probability at least $1-\delta$:
\begin{align}
     D_{H^k}\left(f^*,f\right) \leq D_{h^k}\left(f^*,f\right)+r(|\Omega|,\xi,B,\Tilde{B}, c,I,\delta),
\end{align}
where $r(|\Omega|,\xi,B,\Tilde{B}, c,I,\delta)\to 0$ as $I\to\infty$.
\end{thm}

\begin{thm}[Generalization bound for $L^p$]
\label{thm:sobolev_loss_Lp}

Let $f^*\in \mathcal{F}\subseteq H^k([0,2\pi]^N)$ be a target function, and let there be a $B>0$ and a $\Tilde{B}>0$, such that $\mathcal{F}_{\Omega}^{B}\subseteq \mathcal{H}_{\Omega}^{\Tilde{B}}$ is a suitable model family.
Let us further assume that $ \ell_{h^k}(f_1(\bm{x}),f_2(\bm{x}))\leq c$ for all $\mathbf{x}\in [0,2\pi]^N$, and for all $f_1, f_2 \in\mathcal{F}_{\Omega}^{B}$ or $\mathcal{F}$. 
Assume that $k,p\in\mathbb{N}$ satisfy one of the two following cases:
\begin{enumerate}
    \item $N\left(\frac{1}{2}-\frac{1}{p}\right)<k<N/2$ and $1\leq p<N$.
    \item $k\geq N/2$  and $1\leq p<\infty$.
\end{enumerate}
For any $\delta\in(0,1)$ and the empirical risk $D_{h^k}\left(f^*,f\right)$ trained on an i.i.d. training data $S$ with size $I$ and containing data of $\xi$ partial derivatives, the following holds for all functions $f\in \mathcal{F}_{\Omega}^{B}$ with probability at least $1-\delta$:
\begin{align}
    \dfrac{1}{C}D_{L^p}\left(f^*,f\right) \leq D_{h^k}\left(f^*,f\right)+r(|\Omega|,\xi,B,\Tilde{B}, c,I,\delta),
\end{align}
where $C$ is a constant and $r(|\Omega|,\xi,B,\Tilde{B}, c,I,\delta)\to 0$ as $I\to\infty$.
\end{thm}

\begin{thm}[Generalization bound for $C^0$]\label{thm:sobolev_loss_C0}
Let $f^*\in \mathcal{F}\subseteq H^k([0,2\pi]^N)$ be a target function, and let there be a $B>0$ and a $\Tilde{B}>0$, such that $\mathcal{F}_{\Omega}^{B}\subseteq \mathcal{H}_{\Omega}^{\Tilde{B}}$ is a suitable model family.
Let us further assume that $ \ell_{h^k}(f_1(\bm{x}),f_2(\bm{x}))\leq c$ for all $\mathbf{x}\in [0,2\pi]^N$, and for all $f_1, f_2 \in\mathcal{F}_{\Omega}^{B}$ or $\mathcal{F}$ and that $\|f\|_{\infty}\leq B$ for all $f \in\mathcal{F}_{\Omega}^{B}$.
Assume, that $k\in\mathbb{N}$ satisfies $k>N/2$.
For any $\delta\in(0,1)$ and the empirical risk $D_{h^k}\left(f^*,f\right)$ trained on an i.i.d. training data $S$ with size $I$ and containing data of $\xi$ partial derivatives, the following holds for all functions $f\in \mathcal{F}_{\Omega}^{B}$ with probability at least $1-\delta$:
\begin{align}
   \frac{1}{C}D_{C^0}\left(f^*,f\right) \leq D_{h^k}\left(f^*,f\right)+r(|\Omega|,\xi,B,\Tilde{B}, c,I,\delta),
\end{align}
where $C$ is a constant and $r(|\Omega|,\xi,B,\Tilde{B}, c,I,\delta)\to 0$ as $I\to\infty$.
\end{thm}
The proofs of Theorems \ref{thm:sobolev_loss_Hk}, \ref{thm:sobolev_loss_Lp} and \ref{thm:sobolev_loss_C0} can be found in Appendix \ref{appendix:appendix2}.\\
A consequence of Theorem \ref{thm:sobolev_loss_C0} is that, if the order of the derivatives that we have at our disposal are higher than half the number of input dimensions ($k> N/2$), our solution of the $D_{h^k}$ problem is also a solution of the $D_{C^0}$ problem, corresponding to uniform convergence. 
It means that training with the $\ell_{h^k}$ loss function (for $k> N/2$), which sums the $\ell_{l^2}$ losses of function and derivative values, is sufficient for an approximation in $C^0$. This would not be possible by a training with $\ell_{l^2}$ loss function and more practical than the training with the $\ell_{l^{\infty}}$ loss function, as described above.\\
Note that we face a curse of dimensionality-like phenomenon as the dimension of the input grows.
In this case, the number of terms that go into the  $\ell_{h^k}$ loss function grows exponentially with $k$, as we have to take into account mixed derivatives.
Hence, for high dimensional problems the demand on data of partial derivatives is higher and only if they are available, this generalization bound holds.\\
Further, the requirement of quantum resources for evaluating $D_{h^1}(f^*,f)$ is higher than for the evaluation of $ D_{l^2}(f^*,f)$. If we use the parameter shift rule for the evaluation of the derivatives, we need to evaluate $I(1+2N)$ different PQCs. Similar to the demand on training data, this number of PQCs to evaluate $D_{h^k}(f^*,f)$ grows exponentially in $k$. However, even if the amount of training data is the same (and implying an increase of required PQC evaluations up to a factor of $2$), the training with the $\ell_{h^k}$ loss function shows the promised advantages, as presented in \cite{huge2020differential}.\\
The last property we wish to highlight is the fact that the generalization bounds connect the empirical risk with the full risk, but they do not give us information of whether they can both tend to zero. 
In order to tackle that question we need to combine the results of the theorems present in this section with the ones present in Section \ref{sec:data_normalization}. 
For example, if we try to fit a one-dimensional function which is not periodic on $[0,2\pi]$, using model families $\mathcal{F}_{\Omega}^{B}$ and $\mathcal{H}_{\Omega}^{\tilde{B}}$ and the $\ell_{h^1}$ loss function, as we increase the number of sample points both sides of the inequality will tend to the same constant but they will not converge to zero. 
In this regard, observe that the fact that the empirical risk goes to zero is a sufficient but not a necessary condition for the target distance to also tend to zero. Following the same example, if instead of training the model using the $\ell_{h^1}$ loss function we trained the model using the $\ell_{l^2}$ loss function, then the $L^2$ distance will vanish.
This idea is illustrated in Figure  \ref{fig:f_DML}.
The bottom line is that more information in the training data does not always equate to a better approximation, if we are not very careful with the necessary data normalization.\\

\section{Numerical experiments}\label{sec:numerical_experiments}

In this section we illustrate the theoretical discussion of Sections \ref{sec:data_normalization} and \ref{sec:differential_machine_learning} with an illustrative example: the approximation of function $f^* = \frac{x}{2\pi}, \, x\in [-\pi,\pi]$ by the PQC in Figure \ref{fig:PQC_arquitecture}:
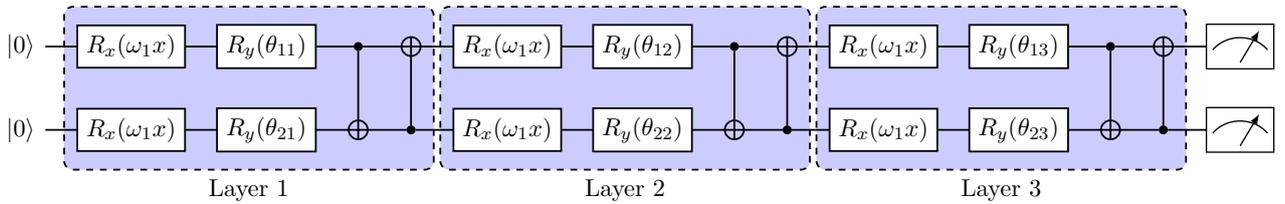
\begin{figure*}
    \centering
    \resizebox{\textwidth}{!}{
    \begin{quantikz}
    \lstick{$\ket{0}$} 
    & \gate{R_x(\omega_1 x)}\QLayer{2}{4}{Layer 1} & \gate{R_y(\theta_{11})} & \ctrl{1} &\targ{}
    & \gate{R_x(\omega_1 x)}\QLayer{2}{4}{Layer 2} & \gate{R_y(\theta_{12})} & \ctrl{1} &\targ{}
    & \gate{R_x(\omega_1 x)}\QLayer{2}{4}{Layer 3} & \gate{R_y(\theta_{13})} & \ctrl{1} &\targ{}
    & \meter{}\\
    \lstick{$\ket{0}$}
    & \gate{R_x(\omega_1 x)}                       & \gate{R_y(\theta_{21})} & \targ{}  &\ctrl{-1}
    & \gate{R_x(\omega_1 x)}                       & \gate{R_y(\theta_{22})} & \targ{}  &\ctrl{-1}
    & \gate{R_x(\omega_1 x)}                       & \gate{R_y(\theta_{23})} & \targ{}  &\ctrl{-1}
    & \meter{}
    \end{quantikz}
    }
    \caption{Architecture $U\left(x,\bm{\theta}\right)$ used in the experiments. The parameters $\theta_{ij}$ are variational parameters. Each qubit is measured in the Pauli-$Z$ basis.}
    \label{fig:PQC_arquitecture}
\end{figure*}
We conduct two different numerical experiments and show them in Figures \ref{fig:f_pqc} and \ref{fig:f_DML}. 
We chose a linear function to show that even in this simple case, the numerical tests fail utterly if the results of Sections \ref{sec:data_normalization} and \ref{sec:differential_machine_learning} are not applied.\\
All simulations have been performed using $10$ points  ($10$ for the labels plus $10$ for the derivative values when they are present) uniformly distributed along the domain for the training phase.
Each experiment has been repeated $100$ times and we depict the $25$, $50$ and $75$ percentiles in colored solid lines in Figure \ref{fig:f_pqc}.
The legends call the result of the PQCs as $f_{\bullet}(\cdot)$, where the subscript denotes under which loss function we have done the training and in the parentheses we indicate which normalization we have chosen.\\
In Figure \ref{fig:f_pqc} we compare the performance of our PQC under different normalizations.
We normalize the data to lie in the domains $\left[-\frac{\pi}{2},\frac{\pi}{2}\right]$, $[-\pi,\pi]$ and $[-2\pi,2\pi]$, respectively. 
When we normalized our data to lie in the range $[-\frac{\pi}{2},\frac{\pi}{2}]$ we get the best results,  as we expected due to Theorem \ref{thm:C_0_approximation}.\\
In contrast, when the data is normalized to lie in the range $[-2\pi,2\pi]$ we obtain very poor approximation results, because in this case, it is not possible to approximate with the $C^0$-distance or the $L^2$-distance.
The intermediate regime happens when we normalize the data to lie in the range $[-\pi,\pi]$, here we obtain a reasonable approximation except for the boundaries. 
This is a consequence of approximating with the $L^2$-distance instead of the $C^0$-distance: we cannot guarantee that the error will be reduced on any given point. 
This behavior remains even when we increase the size of the circuit and the number of given points.\\
\begin{figure*}
    \centering
\begin{tikzpicture}
    \begin{axis}[
      width=0.32\textwidth,
      legend pos=north west,
      xmin = -pi,
      xmax = pi,
      ymax = 0.6,
      ymin = -0.6,
      xlabel={$x$},
      ylabel={$f^*(x) \text{ and } f_{l^2}(x)$}
    ]
    \addplot[color=black]table [x=x,y=y]{data/f_spvsd_half.dat};
    \addplot[name path=mean,color=green]table [x=x,y=y_pred]{data/f_spvsd_half.dat};
    \addplot[name path=upper,color=green!70]table [x=x,y=y_pred_upper]{data/f_spvsd_half.dat};
    \addplot[name path=lower,color=green!70]table [x=x,y=y_pred_lower]{data/f_spvsd_half.dat};
    \addplot[green!50,fill opacity=0.5] fill between[of=lower and upper];

    \legend{$f^*$,$f_{l^2}\left(\left[-\frac{\pi}{2}\comma \frac{\pi}{2}\right]\right)$}
    \end{axis}
\end{tikzpicture}
%%%%%%%%%%%%%%%%%%%%%%%%%%%%%%%%%%%%%%%%%%%%%%%%%%%%%%%%%
\begin{tikzpicture}
    \begin{axis}[
      width=0.32\textwidth,
      legend pos=north west,
      xmin = -pi,
      xmax = pi,
      ymax = 0.6,
      ymin = -0.6,
      xlabel={$x$},
      ylabel={$f^*(x) \text{ and } f_{l^2}(x)$}
    ]
    \addplot[color=black]table [x=x,y=y]{data/f_spvsd_full.dat};
    \addplot[name path=mean,color=yellow]table [x=x,y=y_pred]{data/f_spvsd_full.dat};
    \addplot[name path=upper,color=yellow!70]table [x=x,y=y_pred_upper]{data/f_spvsd_full.dat};
    \addplot[name path=lower,color=yellow!70]table [x=x,y=y_pred_lower]{data/f_spvsd_full.dat};
    \addplot[yellow!50,fill opacity=0.5] fill between[of=lower and upper];

    \legend{$f^*$,$f_{l^2}\left(\left[-\pi\comma \pi\right]\right)$}
    \end{axis}
\end{tikzpicture}
%%%%%%%%%%%%%%%%%%%%%%%%%%%%%%%%%%%%%%%%%%%%%%%%%%%%%%%%%
\begin{tikzpicture}
    \begin{axis}[
      width=0.32\textwidth,
      legend pos=north west,
      xmin = -pi,
      xmax = pi,
      ymax = 0.6,
      ymin = -0.6,
      xlabel={$x$},
      ylabel={$f^*(x) \text{ and } f_{l^2}(x)$}
    ]
    \addplot[color=black]table [x=x,y=y]{data/f_spvsd_double.dat};
    \addplot[name path=mean,color=red]table [x=x,y=y_pred]{data/f_spvsd_double.dat};
    \addplot[name path=upper,color=red!70]table [x=x,y=y_pred_upper]{data/f_spvsd_double.dat};
    \addplot[name path=lower,color=red!70]table [x=x,y=y_pred_lower]{data/f_spvsd_double.dat};
    \addplot[red!50,fill opacity=0.5] fill between[of=lower and upper];

    \legend{$f^*$,$f_{l^2}\left(\left[-2\pi\comma 2\pi\right]\right)$}
    \end{axis}
\end{tikzpicture}
%%%%%%%%%%%%%%%%%%%%%%%%%%%%%%%%%%%%%%%%%%%%%%%%%%%%%%%%%
\caption{In this picture we have trained the PQC of Figure \ref{fig:PQC_arquitecture} to approximate the function $f^* = \frac{x}{2\pi}$. We have used $10$ training points, the $\ell_{l^2}$ loss function and $100$ epochs with the Adam optimizer.
The experiments have been repeated $100$ times.
In the left panel we have normalized the data to lie in the interval $\left[-\frac{\pi}{2},\frac{\pi}{2}\right]$.
In the central panel we have normalized the data to lie in the interval $\left[-\pi,\pi\right]$.
In the right panel we have normalized the data to lie in the interval $\left[-2\pi,2\pi\right]$.}
\label{fig:f_pqc}
\end{figure*}
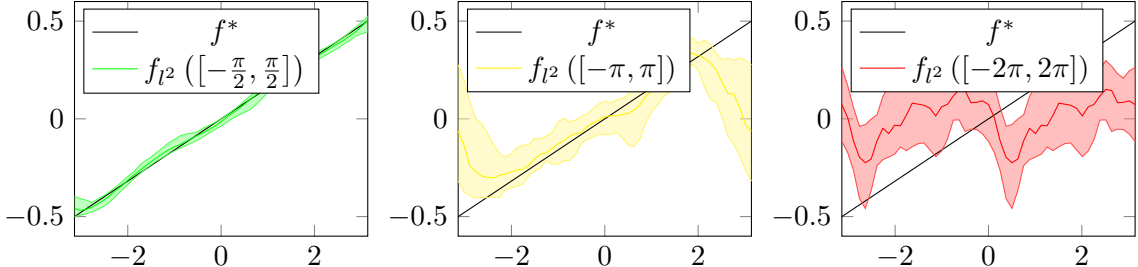

In Figure \ref{fig:f_DML}, we study the impact of the different loss functions with different normalizations in the learning problem.
We simulated the regression using two different loss functions, $\ell_{h^1}$ and $\ell_{l^2}$ under two different normalizations, with the domains $\left[-\frac{\pi}{2},\frac{\pi}{2}\right]$ and $\left[-\pi,\pi\right]$. 
The first noticeable phenomenon that we can see is that using the $h^1$ norm instead of the $l^2$ norm when the data is normalized to lie in the interval $\left[-\frac{\pi}{2},\frac{\pi}{2}\right]$ not only reduces the variance stemming from repeating the experiments 100 times, but also has some impact on the bias.
What might be more surprising is the effect of the $h^1$ norm when the data is normalized to lie in the interval $\left[-\pi,\pi\right]$.
Instead of getting a better approximation w.r.t. the $l^2$ we worsen it.
We explain it with the fact that, when we normalize the data to lie in the interval $\left[-\pi,\pi\right]$, our PQC is not an approximator of $H^1$ but it is an approximator of $L^2$, i.e., it can approximate the function but it cannot simultaneously approximate the function and the derivatives.
Thus, in the minimization process the PQC tries to find a balance between the error in the function and the error in the derivatives, worsening the results with respect to the quality of the function approximation.\\

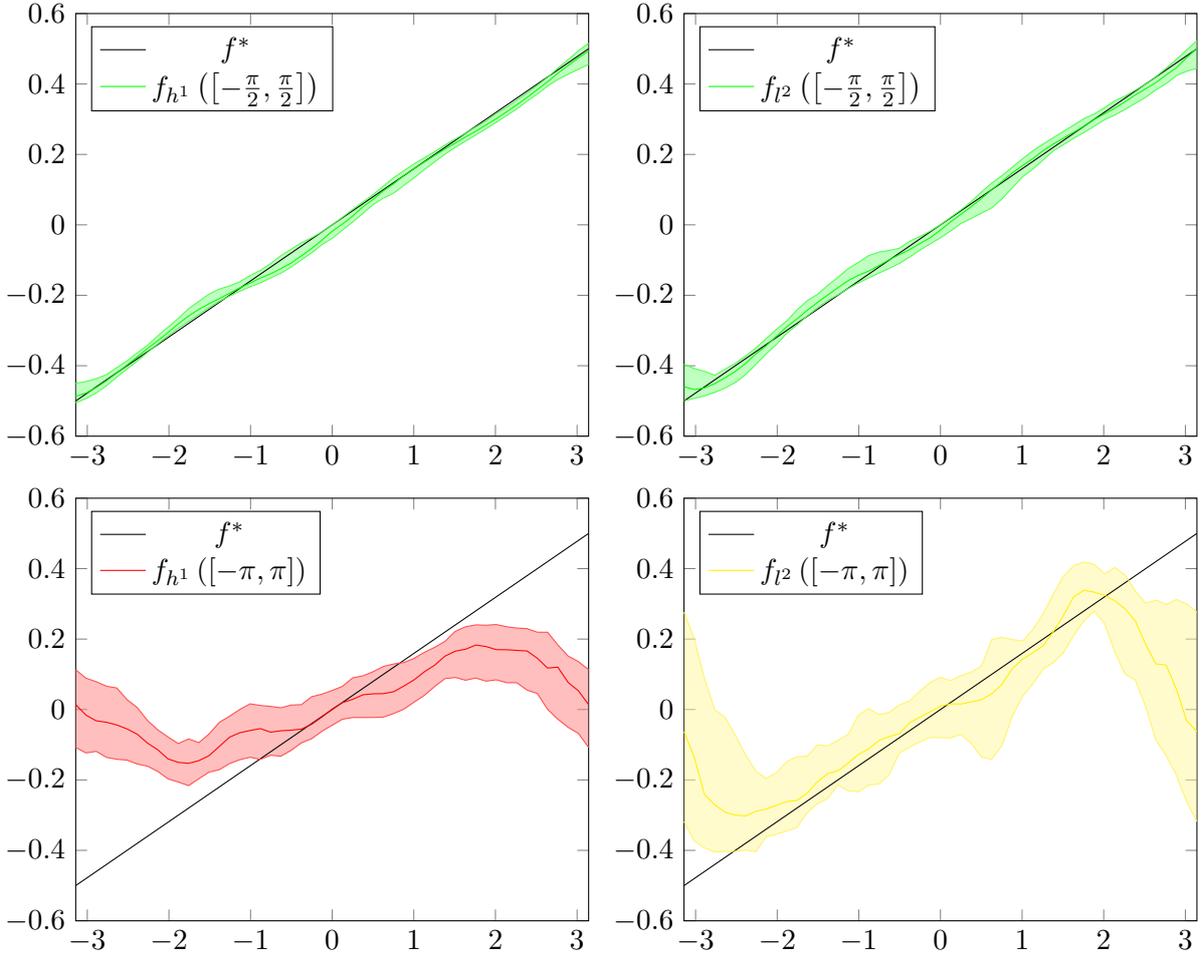
\begin{figure*}
    \centering
    \begin{tikzpicture}
    \begin{axis}[
      width=0.45 \textwidth,
      legend pos=north west,
      xmin = -pi,
      xmax = pi,
      ymax = 0.6,
      ymin = -0.6,
      xlabel={$x$},
      ylabel={$f^*(x) \text{ and } f_{h^1}(x)$}
    ]
    \addplot[color=black]table [x=x,y=y]{data/f_spvsd_half.dat};
    \addplot[name path=mean,color=green]table [x=x,y=y_pred]{data/f_gradient_half.dat};
    \addplot[name path=upper,color=green!70]table [x=x,y=y_pred_upper]{data/f_gradient_half.dat};
    \addplot[name path=lower,color=green!70]table [x=x,y=y_pred_lower]{data/f_gradient_half.dat};
    \addplot[green!50,fill opacity=0.5] fill between[of=lower and upper];

    \legend{$f^*$,$f_{h^1}\left(\left[-\frac{\pi}{2}\comma \frac{\pi}{2}\right]\right)$}
    \end{axis}
\end{tikzpicture}
%%%%%%%%%%%%%%%%%%%%%%%%%%%%%%%%%%%%%%%%%%%%%%%%%%%%%%%%%
\begin{tikzpicture}
    \begin{axis}[
      width=0.45\textwidth,
      legend pos=north west,
      xmin = -pi,
      xmax = pi,
      ymax = 0.6,
      ymin = -0.6,
      xlabel={$x$},
      ylabel={$f^*(x) \text{ and } f_{l^2}(x)$}
    ]
    \addplot[color=black]table [x=x,y=y]{data/f_spvsd_half.dat};
    \addplot[name path=mean,color=green]table [x=x,y=y_pred]{data/f_spvsd_half.dat};
    \addplot[name path=upper,color=green!70]table [x=x,y=y_pred_upper]{data/f_spvsd_half.dat};
    \addplot[name path=lower,color=green!70]table [x=x,y=y_pred_lower]{data/f_spvsd_half.dat};
    \addplot[green!50,fill opacity=0.5] fill between[of=lower and upper];

    \legend{$f^*$,$f_{l^2}\left(\left[-\frac{\pi}{2}\comma \frac{\pi}{2}\right]\right)$}
    \end{axis}
\end{tikzpicture}
%%%%%%%%%%%%%%%%%%%%%%%%%%%%%%%%%%%%%%%%%%%%%%%%%%%%%%%%%
\begin{tikzpicture}
    \begin{axis}[
      width=0.45\textwidth,
      legend pos=north west,
      xmin = -pi,
      xmax = pi,
      ymax = 0.6,
      ymin = -0.6,
      xlabel={$x$},
      ylabel={$f^*(x) \text{ and } f_{h^1}(x)$}
    ]
    \addplot[color=black]table [x=x,y=y]{data/f_gradient_full.dat};
    \addplot[name path=mean,color=red]table [x=x,y=y_pred]{data/f_gradient_full.dat};
    \addplot[name path=upper,color=red!70]table [x=x,y=y_pred_upper]{data/f_gradient_full.dat};
    \addplot[name path=lower,color=red!70]table [x=x,y=y_pred_lower]{data/f_gradient_full.dat};
    \addplot[red!50,fill opacity=0.5] fill between[of=lower and upper];

    \legend{$f^*$,$f_{h^1}\left(\left[-\pi\comma \pi\right]\right)$}
    \end{axis}
\end{tikzpicture}
%%%%%%%%%%%%%%%%%%%%%%%%%%%%%%%%%%%%%%%%%%%%%%%%%%%%%%%%%
\begin{tikzpicture}
    \begin{axis}[
      width=0.45\textwidth,
      legend pos=north west,
      xmin = -pi,
      xmax = pi,
      ymax = 0.6,
      ymin = -0.6,
      xlabel={$x$},
      ylabel={$f^*(x) \text{ and } f_{l^2}(x)$}
    ]
    \addplot[color=black]table [x=x,y=y]{data/f_spvsd_full.dat};
    \addplot[name path=mean,color=yellow]table [x=x,y=y_pred]{data/f_spvsd_full.dat};
    \addplot[name path=upper,color=yellow!70]table [x=x,y=y_pred_upper]{data/f_spvsd_full.dat};
    \addplot[name path=lower,color=yellow!70]table [x=x,y=y_pred_lower]{data/f_spvsd_full.dat};
    \addplot[yellow!50,fill opacity=0.5] fill between[of=lower and upper];

    \legend{$f^*$,$f_{l^2}\left(\left[-\pi\comma \pi\right]\right)$}
    \end{axis}
\end{tikzpicture}
\caption{In this picture we have trained the PQC of Figure \ref{fig:PQC_arquitecture} to approximate the function $f^* = \frac{x}{2\pi}$, using the two different loss functions $\ell_{h^1}$ and $\ell_{l^2}$.
We have used $10$ training points ($10$ for the labels plus $10$ for the derivative values when they are present) and $100$ epochs with the Adam optimizer.
The experiments have been repeated $100$ times.
In the upper panel, we have normalized the data to lie in the interval $\left[-\frac{\pi}{2},\frac{\pi}{2}\right]$.
In the lower panel we have normalized the data to lie in the interval $\left[-\pi,\pi\right]$.}
\label{fig:f_DML}
\end{figure*}

\section{Conclusions}\label{sec:conclusions}
In this work, we have developed a broader theory of approximation capacities of PQCs. We have shown how an appropriate choice of the data normalization greatly improves the expressivity of the PQCs.
More specifically, we showed that a min-max feature scaling that normalizes the input data along each dimension to lie in the range $[-\frac{\pi}{2},\frac{\pi}{2}]$ makes PQCs universal approximators in the $L^p$ space with $1\leq p< \infty$, the continuous function space and the $H^k$ space.\\
Moreover, since with this normalization we are able to approximate functions in the sense of the $L^p$, the $C^0$ and the $H^k$ distance, we discussed that a loss function which is consistent with those distances in training the models might be more appropriate than other choices.
In particular, the natural choice for the $C^0$ would be the $l^\infty$ distance.
However, since the $l^\infty$ distance is not differentiable, which makes the optimization of PQCs harder, we leveraged Sobolev inequalities to show that the $h^1$ distance is consistent with the $C^0$ distance in $\mathbb{R}$ while being differentiable. We showed further, that the  the $h^k$ distances are consistent with the $L^p$ and the $H^k$ distances.\\
Lastly, we performed some numerical experiments to illustrate how this simple choice of normalization and loss function can vastly improve the results in practice.\\ 

The data normalization technique can be seen as a complementary result to the work of \cite{schuld2021effect}.
Nevertheless, there is still much work to do in this direction.
For example, if instead of only taking a min-max feature scaling, we can combine it with a mapping of the form $\tilde{x} = \arcsin(x)$ to end up with a series that closely resembles Chebyshev polynomials, which are better suited for certain problems.
In analogy with neural networks, the data encoding strategy is playing a similar role to that of the activation functions.\\ 

The relation between the $\ell_{h^k}$ loss functions and the $L^p$ generalization bounds can be seen as a complementary result to differential machine learning \cite{huge2020differential} and to generalization bounds for PQCs as derived in \cite{Caro2021encodingdependent}.
This is the first work that gives some insight on why differential machine learning leads to better generalization results. 
From the relations that we derived, one would expect this technique to fail as we increase the input dimension.
However, in practice it has demonstrated very good results, as shown in \cite{huge2020differential}, where a 7-dimensional Basket option was trained using the $\ell_{h^1}$ loss function.
An interesting line of research would be to study the threshold at which differential machine learning starts to fail.\\
Since a natural application are physical systems governed by differential equations
where data on the derivatives of a target function are available, another open question remains regarding how our approach compares to standard differential equation solvers in these scenarios.

\section{Acknowledgments}
The authors would like to thank Adrián Pérez Salinas for useful feedback on an earlier version of this paper and Hao Wang and Carlos Vázquez for helpful discussions.\\
JT, VD and DD acknowledge the support received by the Dutch National Growth Fund (NGF), as part of the Quantum Delta NL programme.\\
JT acknowledges the support received from the European Union’s Horizon Europe research and innovation programme through the ERC StG FINE-TEA-SQUAD (Grant No. 101040729).'\\
VD and AM acknowledge the support by the project
NEASQC funded from the European Union’s Horizon
2020 research and innovation programme (grant agree-
ment No 951821).\\
VD acknowledges by the Dutch Research Council
(NWO/OCW), as part of the Quantum Software Consortium programme (project number 024.003.037).\\
AM acknowledges the support received from the Centro de Investigación de Galicia ``CITIC", funded by Xunta de Galicia and the European Union (European Regional Development Fund- Galicia 2014-2020 Program), by grant ED431G 2019/01.\\
The views and opinions expressed here are solely those of the authors and do not necessarily reflect those of the funding institutions. Neither of the funding institution can be held responsible for them.

\bibliographystyle{plainnat}
\bibliography{mainbib}

\onecolumn
\appendix

\section{Proof of Theorems \ref{thm:L_p_approximation}, \ref{thm:C_0_approximation} and \ref{thm:H^k_approximation}}\label{appendix:fourier_series_uniform_approximation}
For proving Theorems \ref{thm:L_p_approximation}, \ref{thm:C_0_approximation} and \ref{thm:H^k_approximation}, we need two preliminary results. 
Firstly, we need to show that a quantum circuit can realize the $\ell^1$-Fejér's mean of $C^0\left(\mathbb{T}^N\right)$ and $L^p\left(\mathbb{T}^N\right), \, \forall \, 1\leq p< \infty$ functions.
Secondly, we need to prove that we can define periodic extensions of functions belonging to $C^0\left(U\right)$ and $H^k\left(U\right), \, \forall \, 1\leq k< \infty$,  where $U$ is compactly contained in $\mathbb{T}^N$
to functions belonging to $C^0\left(\mathbb{T}^N\right)$ and $H^k\left(\mathbb{T}^N\right), \, \forall \, 1\leq k< \infty$ respectively . 
The combination of both results plus Fejér's theorem in multiple dimensions naturally yields Theorems \ref{thm:L_p_approximation} and \ref{thm:C_0_approximation}. Theorem \ref{thm:H^k_approximation} can be proven by a standard approximation Theorem of the Fourier series.

\subsection{Féjer's mean}\label{appendix:Fejer_mean}
We call the function
\begin{align}\label{eq:fejersmean}
    \sigma_{NK}(f_{m'})=\sum_{\mathbf{j}\in\mathbb{Z}_K^N}\left(1-\dfrac{\norm{\mathbf{j}}_1}{NK}\right)\hat{f}_{\mathbf{j}}e^{i\mathbf{x}\cdot\mathbf{j}}\ ,
\end{align} 
where $\hat{f}_{\mathbf{j}}$ is the $\mathbf{j}$-th Fourier coefficient of $f_{m'}$, the $\ell^1$-Fejér's mean of $f_{m'}$.\\

 We will show that our PQC can realize the Fejér's mean of any well-defined function.
In Appendix C of \cite{schuld2021effect}, the authors showed that the quantum model
family $f_{m'}$ can be written as a generalized trigonometric series of the form
\begin{align}
\label{eq:quantummodelasfourierappendix_Fejersmean}
    f_{m'}(\mathbf{x})=\sum_{\mathbf{j}\in\mathbb{Z}_K^N}c_{\mathbf{j}}e^{i\mathbf{x}\cdot\mathbf{j}}\ ,
\end{align}
where $\mathbb{Z}_K^N=\{-K,-K+1,...,0,...,K-1,K\}^N$ is contained in the Cartesian product of the frequency spectrum associated with $H_m$, as defined in Definition \ref{definition: universal hamiltonian family} and that the coefficients $c_{\mathbf{j}}$ are completely determined by the observable  freely up to the complex-conjugation symmetry that guarantees that the model output is a real-valued function. Note that we can choose the coefficients $c_{\mathbf{j}}$ as:
\begin{align}
    c_{\mathbf{j}}=\left(1-\dfrac{\norm{\mathbf{j}}_1}{NK}\right)\hat{f}_{\mathbf{j}},
\end{align}
which are the coefficients of the $\ell^1$-Fejér's mean in Equation \eqref{eq:fejersmean}.

\subsection{\texorpdfstring{Periodic extension for $C^0$ functions}{Periodic extension for functions}}\label{appendix:periodic_extension}
By the Tietze extension theorem \cite{royden1968real}, there exists a function $g_1\in C^0(\mathbb{R}^N)$ with $g_1|_{U}=f^*$.
Then, we define a function $g_2\in C^0(\mathbb{R}^N)$ with $g_2|_{\overline{U}}=1$ and $g_2|_{\mathbb{R}^N\backslash V}=0$, where $V$ is defined as $\overline{U}\subset V\subset (0,2\pi)^N$. This set $V$ exists since $U$ is compactly contained in $[0,2\pi]^N$.\\ 
We can explicitly construct the function $g_2$ in the following way:
Let $\delta>0$, such that the closure $\overline{\omega_{2\delta}}$ of the $2\delta$-neighborhood of $\omega$, is contained in $[0,2\pi]^N$, which is possible due to $U$ being compactly contained in $[0,2\pi]^N$. 
We define $V:=\omega_{2\delta}$ and a function $\psi_{\delta}\in C^0(\mathbb{R}^N)$, supported on the $\delta$ Ball in $\mathbb{R}^N$ centered around $0$ and normalized as $\int_{\mathbb{R}^N}\psi_{\delta}(x)dx=1$. Then, we define $g_2$ as the convolution of $\mathds{1}_{U_\delta}$ and $\psi_{\delta}$:
\begin{align}
    g_2(x)=\int_{\mathbb{R}^N}\mathds{1}_{U_\delta}(\tau)\psi_{\delta}(\tau-x)d\tau\ .
\end{align}
With this construction, $g_2$ satisfies the asked properties.
We define the extension $f_{ext}$ as the product $g_1g_2$, which yields a function $f^*_{ext}$ with 
\begin{align}
    f^*_{ext}|_{U}&=f^*\ ,\\
    f^*_{ext}|_{\mathbb{R}^N\backslash
V}&=0\ ,\text{ hence}\\
f^*_{ext}(x)&=f^*_{ext}(y)\quad\forall x,y\in\partial \mathbb{T}^N\ .
\end{align}
The such defined extension $f^*_{ext}$ is thus an element of $C^0([0,2\pi]^N)$ with periodic boundary conditions, so we can map it onto the $N$-dimensional torus $\mathbb{T}^N$. 

\subsection{\texorpdfstring{Periodic extension for $H^k$ functions}{Periodic extension for functions}}\label{appendix:periodic_extension_Hk}
By the extension theorems for Sobolev functions \cite[Theorem 2.2, Part 2]{burenkov1999extension}, there exists a function $g_1\in H^{k}(\mathbb{R}^N)$ with $g_1|_{U}=f^*$.
Then, we define a function $g_2\in H^{k}(\mathbb{R}^N)$ with $g_2|_{\overline{U}}=1$ and $g_2|_{\mathbb{R}^N\backslash V}=0$, where $V$ is defined as $\overline{U}\subset V\subset (0,2\pi)^N$. This set $V$ exists since $U$ is compactly contained in $[0,2\pi]^N$.\\ 
We can explicitly construct the function $g_2$ in the following way:
Let $\delta>0$, such that the closure $\overline{\omega_{2\delta}}$ of the $2\delta$-neighborhood of $\omega$, is contained in $[0,2\pi]^N$, which is possible due to $U$ being compactly contained in $[0,2\pi]^N$. 
We define $V:=\omega_{2\delta}$ and a function $\psi_{\delta}\in H^{k}(\mathbb{R}^N)$, supported on the $\delta$ Ball in $\mathbb{R}^N$ centered around $0$ and normalized as $\int_{\mathbb{R}^N}\psi_{\delta}(x)dx=1$. Then, we define $g_2$ as the convolution of $\mathds{1}_{U_\delta}$ and $\psi_{\delta}$:
\begin{align}
    g_2(x)=\int_{\mathbb{R}^N}\mathds{1}_{U_\delta}(\tau)\psi_{\delta}(\tau-x)d\tau\ .
\end{align}
With this construction, $g_2$ satisfies the asked properties.
We define the extension $f_{ext}$ as the product $g_1g_2$, which yields a function $f^*_{ext}$ with 
\begin{align}
    f^*_{ext}|_{U}&=f^*\ ,\\
    f^*_{ext}|_{\mathbb{R}^N\backslash
V}&=0\ ,\text{ hence}\\
f^*_{ext}(x)&=f^*_{ext}(y)\quad\forall x,y\in\partial \mathbb{T}^N\ .
\end{align}
The such defined extension $f^*_{ext}$ is thus an element of $H^{k}([0,2\pi]^N)$ with periodic boundary conditions, so we can map it onto the $N$-dimensional torus $\mathbb{T}^N$.

\subsection{Proof of Theorems \ref{thm:L_p_approximation}, \ref{thm:C_0_approximation} and \ref{thm:H^k_approximation}}
The final step leverages the power of Fejèr's theorem in multiple dimensions:
\begin{thm}\cite{WEISZ201199}[Theorem 2]\label{thm:fejer_L_p}
    For all functions $f^*\in L^p\left(\mathbb{T}^N\right)$ with $1\leq p <\infty$, and for all $\epsilon >0$, there exists some $t \in \mathbb{N}$, such that
\begin{equation}
   \norm{\sigma_{t}\left(f\right)-f^*}_{L^p}<\epsilon.
\end{equation}
\end{thm}
Combining Theorem \ref{thm:fejer_L_p} with the fact that quantum circuits can recover any $\ell^1$-Fejér's mean as shown in Appendix \ref{appendix:Fejer_mean} directly implies Theorem \ref{thm:L_p_approximation}.\\ \\
Similarly, for continuous functions we have another version of Fejér's theorem for continuous functions:
\begin{thm}\cite{WEISZ201199}[Theorem 2]\label{thm:fejer_C_0}
    For all functions $f^*\in C^0\left(\mathbb{T}^N\right)$, and for all $\epsilon >0$, there exists some $t \in \mathbb{N}$, such that
\begin{equation}
   \norm{\sigma_{t}\left(f\right)-f^*}_{\infty}<\epsilon.
\end{equation}
\end{thm}
Combining Theorem \ref{thm:fejer_C_0} with the fact that quantum circuits can recover any $\ell^1$-Fejér's mean as shown in Appendix \ref{appendix:Fejer_mean} and the fact that we can extend any function in $C^0\left(U\right)^N, \, \forall \, 1\leq p< \infty$ where $U$ is compactly contained in $\mathbb{T}^N$ to a function in $C^0 \left(\mathbb{T}^N\right), \, \forall \, 1\leq p< \infty$ as shown in Appendix \ref{appendix:periodic_extension} directly implies Theorem \ref{thm:C_0_approximation}.\\

We finally prove Theorem \ref{thm:H^k_approximation}, which uses the setup in \cite{schuld2021effect} as described in section \ref{sec:data_normalization}:
    We note firstly that the quantum model family $f_{m'}$ generates a truncated Fourier series $\tilde{f}$ in the domain $[0,2\pi]^N$ of the form
\begin{align}
\label{eq:quantummodelasfourierappendix_H^k}
    \tilde{f}(\mathbf{x})=\sum_{\mathbf{j}\in\mathbb{Z}_K^N}c_{\mathbf{j}}e^{i\mathbf{x}\cdot\mathbf{j}}\ ,
\end{align}
where $\mathbb{Z}_K^N=\{-K,-K+1,...,0,...,K-1,K\}^N$ is contained in the Cartesian product of the frequency spectrum associated with $H_m$, as defined in Definition \ref{definition: universal hamiltonian family}.
The proof of that is written in Appendix C of \cite{schuld2021effect}.\\
Secondly, we can extend the function $f^*$ defined on $U$ to a periodic function $f^*_{ext}$ on $[0,2\pi]^N$ via the construction shown in Appendix \ref{appendix:periodic_extension_Hk}.
As written in Theorem 1.1 in \cite{canuto1982approximation}, the Fourier series of $f^*_{ext}$, which we can write in the form of equation \ref{eq:quantummodelasfourierappendix_H^k}, converges in the $H^k$-distance to $f^*_{ext}$. 
As $f^*_{ext}(x)=f^*(x)$ for all $x\in U$, the Fourier series of $f^*_{ext}$ converges in the $H^k$-distance  to $f^*$ on $U$. This implies Theorem \ref{thm:H^k_approximation}.

\section{Proof of Theorems \ref{thm:sobolev_loss_Hk}, \ref{thm:sobolev_loss_Lp} and \ref{thm:sobolev_loss_C0}}\label{appendix:appendix2}
In this appendix, we prove Theorems \ref{thm:sobolev_loss_Hk}, \ref{thm:sobolev_loss_Lp} and \ref{thm:sobolev_loss_C0}, for which we need several preliminary definitions and results:
\begin{definition}[$L$-Lipschitz loss function]
    Let $(\mathcal{Y},d_\mathcal{Y})$ be a metric space with metric $d_\mathcal{Y}$ and let $\ell:\mathcal{Y}\times \mathcal{Y}\to \mathbb{R}$ be a loss function. We call it $L-$ Lipschitz with regard to a fixed $y\in \mathcal{Y}$, if there exists a constant $L\geq0$, such that for all $z_1,z_2\in \mathbb{R}$,
    \begin{align}
        d_\mathcal{Y}\left(\ell(y,z_1),\ell(y,z_2)\right)\leq L\left|z_1-z_2\right|\ .
    \end{align}
\end{definition}

\begin{thm}[Generalization bound for general trigonometric series] \cite{Caro2021encodingdependent}[Theorem 11]
\label{thm:Generalization bound for general trigonometric series}
Let $N,I\in \mathbb{N}$. Let $B>0$ and $\Tilde{B}>0$ be such that $\mathcal{F}_{\Omega}^{B}\subseteq \mathcal{H}_{\Omega}^{\Tilde{B}}$, for the function families $\mathcal{F}_{\Omega}^{B}$ and $\mathcal{H}_{\Omega}^{\Tilde{B}}$ as defined in Definition \ref{def:functionfamilies}.
Let $\ell:\mathbb{R}\times\mathbb{R}\to [0,c]$ be a bounded loss function such that $\mathbb{R}\ni z\mapsto \ell(y,z)$ is $L$- Lipschitz for all $y\in\mathbb{R}$. 
For any $\delta\in(0,1)$ and for any probability measure $P$ on $[0,2\pi]^N\times\mathbb{R}$, with probability at least $1-\delta$ over the choice of i.i.d. training data $S\in([0,2\pi]^N\times\mathbb{R})^I$ of size $I$, for every $f\in\mathcal{F}_{\Omega}^{B}$, the generalization error can be upper-bounded as
\begin{align}
    &\int_{[0,2\pi]^N\times \mathbb{R}}\ell(f^*(\bm{x}),f(\bm{x}))dP(\bm(x),f^*(\bm{x}))-\frac{1}{|I|}\sum_{\bm{x}_i,f(\bm{x}_i)\in S}\ell(f^*(\bm{x}_i),f(\bm{x}_i))\\
    &\leq\mathcal{O}\left(BL\sqrt{\frac{|\Omega|(\log(|\Omega|)+\log(\Tilde{B}))}{I}}+c\sqrt{\frac{\log(1/\delta)}{I}}\right)\ ,
\end{align}
for a target function $f^*:[0,2\pi]^N\to \mathbb{R}$\ .
\end{thm}

This theorem is written for loss functions that take two real values as an input, which is the case for most loss functions. We show that the theorem holds as well for the loss function $\ell_{h^k}$:

\begin{lemma}
\label{lemma:Generalizationboundalsovalidforsobolev}
    Theorem \ref{thm:Generalization bound for general trigonometric series} holds as well for the loss function $ \ell_{h^k}:\mathbb{R}^{\binom{N}{k}+1}\times\mathbb{R}^{\binom{N}{k}+1}\to [0,c]$ with $N,k\in\mathbb{N}$ by choosing the frequency set $\Omega$ and the bounds $B$ and $\Tilde{B}$ large enough, such that both the functions $f$ of a considered function family $\mathcal{F}$ and their derivatives $D^{\alpha}f$ for $\abs{\alpha}\leq k$ are contained in the families $\mathcal{F}_{\Omega}^{B}\subseteq\mathcal{H}_{\Omega}^{\Tilde{B}}$. 
\end{lemma}
\begin{proof}

 The proof goes analogous to the proof of Theorem 11 in \cite{Caro2021encodingdependent}. There are two points which require special care: \\
    Firstly, we need to adapt the application of Talagrand's lemma which is used to upper bound the Rademacher complexity.
    Let us use the $\ell_{l^2}$ loss function
    \begin{align}
        \ell_{l^2}(f^*(\bm{x}),f(\bm{x}))=\left(f^*(\bm{x})-f(\bm{x})\right)^2\ ,
    \end{align}
    which is related to the loss function $\ell_{h^k}$ by 
    \begin{align}
         \ell_{h^k}(f^*(\bm{x}),f(\bm{x}))=\sum_{\abs{\alpha}\leq k}\ell_{l^2}\left(D^{\alpha}f^*(\bm{x}),D^{\alpha}f(\bm{x})\right)\ .
    \end{align}
    By using the reverse triangle inequality, we can prove the lipschitzness of the loss function $\ell_{l^2}$, for a fixed $f^*(\bm{x})\in L^2([0,2\pi]^N)$:
\begin{align}
    \Bigg|\ell_{l^2}(f_1(\bm{x}),f^*(\bm{x}))-\ell_{l^2}(f_2(\bm{x}),f^*(\bm{x}))\Bigg|&= \Bigg|\abs{f^*(\bm{x})-f_1(\bm{x}))}^2-\abs{f^*(\bm{x})-f_2(\bm{x}))}^2\Bigg|\\
    &\leq \Bigg|\abs{f^*(\bm{x})-f_1(\bm{x})-\left(f^*(\bm{x})-f_2(\bm{x})\right)}^2\Bigg|\\
    &= \Bigg|\abs{\left(f^*(\bm{x})-f^*(\bm{x})\right)-\left(f_1(\bm{x}))-f_2(\bm{x}))\right)}^2\Bigg|\\
    &=\abs{\left(f_1(\bm{x}))-f_2(\bm{x}))\right)}^2\ .
\end{align}
    Thus, the loss function $\ell_{l^2}$ is $L$-Lipschitz with the Lipschitz constant $L=1$. Note that this is the Lipschitz constant of the loss function $\ell_{l^2}$, which is not related to the Lipschitz constant of functions of the function space $L^2([0,2\pi]^N)$.\\
    Parallel to the proof of Theorem 11 in \cite{Caro2021encodingdependent}, we now define the set
    \begin{align}
        \mathcal{G}=\left\{[0,2\pi]^N\times [0,2\pi]^N\ni(\bm{x},\bm{x})\mapsto \ell_{h^k}(f^*(\bm{x}),f(\bm{x}))\Big|f^*\in H^k([0,2\pi]^N)\text{ and }f\in\mathcal{F}^B_\Omega\right\}
    \end{align}
    We can now upper bound the Rademacher complexity $\hat{\mathcal{R}}_S(\mathcal{G})$ for a training set $S$ with $I$ data points and a target function $f^*$ as
    \begin{align}
        \hat{\mathcal{R}}_S(\mathcal{G})&=\frac{1}{I}\mathbb{E}_{\sigma}\left[\sup_{f\in\mathcal{F}^B_\Omega }\sum_{i=1}^{I}\sigma_i\ell_{h^k}(f(\bm{x}_i),f^*(\bm{x}_i))\right]\\
        &=\frac{1}{I}\mathbb{E}_{\sigma}\left[\sup_{f\in\mathcal{F}^B_\Omega }\sum_{i=1}^{I}\sigma_i\sum_{\abs{\alpha}\leq k}\ell_{l^2}(D^{\alpha}f(\bm{x}_i),D^{\alpha}f^*(\bm{x}_i))\right]\\
        &\leq \frac{1}{I}\mathbb{E}_{\sigma}\left[\sup_{D^{\alpha}f(\bm{x})\in\mathcal{F}^B_\Omega, \abs{\alpha}\leq k}\sum_{i=1}^{I}\sigma_i\sum_{\abs{\alpha}\leq k}\ell_{l^2}(D^{\alpha}f(\bm{x}_i),D^{\alpha}f^*(\bm{x}_i))\right]\\
        &= \sum_{\abs{\alpha}\leq k}\frac{1}{I}\mathbb{E}_{\sigma}\left[\sup_{D^{\alpha}f(\bm{x})\in\mathcal{F}^B_\Omega}\sum_{i=1}^{I}\sigma_i\ell_{l^2}(D^{\alpha}f(\bm{x}_i),D^{\alpha}f^*(\bm{x}_i))\right]\\
        &\leq \xi \sup_{\abs{\alpha}\leq k} \frac{1}{I}\mathbb{E}_{\sigma}\left[ \sup_{D^{\alpha}f(\bm{x})\in\mathcal{F}^B_\Omega}\sum_{i=1}^{I}\sigma_i\ell_{l^2}(D^{\alpha}f(\bm{x}_i),D^{\alpha}f^*(\bm{x}_i))\right]\ .
    \end{align}
    The i.i.d. random variables $\sigma_i\in\{-1,1\}$ are the Rademacher random variables and $\xi$ is the number of derivatives $D^{\alpha}$ with $\abs{\alpha}\leq k$.
    Here, we first used the relation between the loss functions $\ell_{l^2}$ and $\ell_{h^k}$. 
    Then, we used the fact that the supremum over functions and derivatives $D^{\alpha}f(\bm{x})\in\mathcal{F}^B_\Omega, \abs{\alpha}\leq k$ which are independent from each other is larger than the supremum which is only taken over the functions $f\in\mathcal{F}^B_\Omega$, in which case the derivatives that are taken account in the loss functions have to be the derivatives of these functions. 
    In the last inequality, we used that each of the $\xi$ terms in the sum $\sum_{\abs{\alpha}\leq k}$ can be upper bounded by its supremum.\\
    We can now apply Talagrand's lemma on the quantity $\frac{1}{I}\mathbb{E}_{\sigma}\left[ \sup_{D^{\alpha}f(\bm{x})\in\mathcal{F}^B_\Omega}\sum_{i=1}^{I}\sigma_i\ell_{l^2}(D^{\alpha}f(\bm{x}_i),D^{\alpha}f^*(\bm{x}_i))\right]$, for a fixed $\abs{\alpha}\leq k$ in which way we obtain the upper bound
    \begin{align}
        \hat{\mathcal{R}}_S(\mathcal{G})\leq \xi \hat{\mathcal{R}}_{S|_x}(\mathcal{F}^B_\Omega)\ ,
    \end{align}
    where we used that the loss function $\ell_{l^2}$ has the Lipschitz constant $L=1$ and where $S|_x:=\{\bm{x}_i\}_{i=0}^{I}$ is the set of the unlabeled training data points.
    The supremum $\sup_{\abs{\alpha}\leq k}$ can be omitted on the right hand side of the bound, since the subset $S|_x$ of the training set does not include the labels $D^{\alpha}f^*(\bm{x}_i)$ and since we assumed the function family $\mathcal{F}^B_\Omega$ to contain the relevant derivatives as well.
    This upper bound corresponds to equation (97) in the proof of Theorem 11 in \cite{Caro2021encodingdependent}, apart from the additional factor $\xi$.\\
    Secondly, in the last step of the proof in \cite{Caro2021encodingdependent}, the authors use standard generalization bounds as stated in Theorem 1.15 in \cite{wolf2023}.
    The formulation of this standard generalization bound theorem allows for the loss function $\ell_{h^k}$ as well.
\end{proof}

\begin{definition}[Compact embedding] \cite{adams2003sobolev}[Definition 1.25]
\label{def:compact_embedding}
Let $X$ and $Y$ be normed spaces with the norms $\norm{\cdot}_X$ and $\norm{\cdot}_Y$, respectively, and $X$ a subspace of $Y$. Let $I:X\to Y$, $Ix=x$ for all $x\in X$ be the embedding operator from $X$ to $Y$.  We say that $X$ is continuously embedded in $Y$, and write $X\to Y$, if there exists a constant $C$, such that
\begin{align}
    \norm{Ix}_Y\leq C\norm{x}_X, \forall x\in X\ .
\end{align}
We call the embedding compact, if $X$ is continuously embedded in $V$ and the embedding operator $I$ is compact.
\end{definition}

\begin{definition}
    We write $C_B^0(U)$ for the space of bounded, continuous functions on $U$.
\end{definition}

\begin{definition}[Finite cone and Cone condition]\cite{adams2003sobolev}[Definitions 4.4 and 4.6]
Let $v,x\in\mathbb{R}^N$ be nonzero vectors, let $\angle(x,v)$ be the angle between vectors $x$ and $v$. For given such $v$, a $\rho>0$ and a $\kappa$ such that $0< \kappa\leq \pi$, the set
\begin{align}
C_{v,\rho,\kappa}=\{x\in\mathbb{R}^N:x=0\text{ or }0<\|x\|\leq\rho,\angle(x,v)\leq\kappa/2\}
\end{align}
is called a finite cone of height $\rho$, axis direction $v$ and aperture angle $\kappa$ with vertex at the origin.\\
We say that $U\subseteq\mathbb{R}^N$ satisfies the cone condition, if there exists a finite cone $C$ such that every $x\in U$ is the vertex of a finite cone $C_x$ contained in $U$ and congruent to $C$.
\end{definition}
\begin{thm}[Rellich-Kondrachov]\cite{adams2003sobolev}[Theorem 6.3, Part I and II]
\label{thm:Rellich-Kondrachov}
    Let $U$ be a domain in $\mathbb{R}^N$ satisfying the cone condition, let $U_0$ be a bounded subdomain of $U$, and let $U_0^N$ be the intersection of $U_0$ with a $N$-dimensional plane in $\mathbb{R}^N$. Let $k\geq 1$ be integers.
    Let one of the following cases hold:
    \begin{enumerate}
        \item $2k< N$ and  $1\leq p< 2N/(N-2k)$
        \item $2k=N$ and $1\leq p<\infty$\
        \item $2k> N$ and $1\leq p<\infty$
    \end{enumerate}
    Then, the following embeddings are compact:
\begin{align}
    H^{k}(U)\to L^{p}(U_0^N)\ .
\end{align}
Additionally, in case 3, the following embedding is compact:
\begin{align}
    H^{k}(U)\to C_B^{0}(U_0^N)\ .
\end{align}

\end{thm}
\begin{remark}
    The theorem relates to the Rellich-Kondrachov Theorem stated in \cite{adams2003sobolev} in the following way: 
    \begin{itemize}
        \item Case 1 and Case 2 are the two cases stated in Part 1 of Theorem 6.3 in \cite{adams2003sobolev}.
        \item Case 3 corresponds to the first and second case of Part 2 in Theorem 6.3 in \cite{adams2003sobolev}.
        \item We use a different notation: The symbols $\Omega,j,p,q,k,n,m$ used in \cite{adams2003sobolev} are here equal to $U,0,2,p,N,N,k$, respectively.
        \item We formulate the theorem for the special cases $W^{k,2}=H^k$ and $W^{0,p}=L^p$ of the Sobolev spaces.
    \end{itemize}
\end{remark}

With these preliminary results, we can prove Theorems \ref{thm:sobolev_loss_Hk}, \ref{thm:sobolev_loss_Lp} and \ref{thm:sobolev_loss_C0}, which we restate here:

\begin{customthm}{5}[Generalization bound for $H^k$]

Let $f^*\in \mathcal{F}\subseteq H^k([0,2\pi]^N)$ be a target function, and let there be a $B>0$ and a $\Tilde{B}>0$, such that $\mathcal{F}_{\Omega}^{B}\subseteq \mathcal{H}_{\Omega}^{\Tilde{B}}$ is a suitable model family.
Let us further assume that $\ell_{h^k}(f_1(\bm{x}),f_2(\bm{x}))\leq c$ for all $\mathbf{x}\in [0,2\pi]^N$, and for all $f_1, f_2 \in\mathcal{F}_{\Omega}^{B}$ or $\mathcal{F}$. 
For any $\delta\in(0,1)$ and the empirical risk $D_{h^k}\left(f^*,f\right)$ trained on an i.i.d. training data $S$ with size $I$ and containing data of $\xi$ partial derivatives, the following holds for all functions $f\in \mathcal{F}_{\Omega}^{B}$ with probability at least $1-\delta$:
\begin{align}
     D_{H^k}\left(f^*,f\right) \leq D_{h^k}\left(f^*,f\right)+r(|\Omega|,\xi,B,\Tilde{B}, c,I,\delta),
\end{align}
where $r(|\Omega|,\xi,B,\Tilde{B}, c,I,\delta)\to 0$ as $I\to\infty$.
\end{customthm}

\begin{proof}
  In the work \cite{Caro2021encodingdependent}, the authors developed generalization bounds for the function family defined in \eqref{eq:family_of_models_dep_on_Hamiltonian_family}. We restated the theorem in Theorem \ref{thm:Generalization bound for general trigonometric series}. As we have shown in Corollary \ref{lemma:Generalizationboundalsovalidforsobolev}, the theorem also holds for the loss function $\ell_{h^k}$.
 
According to the assumption, the function $f^*$ is in $\mathcal{F}_{\Omega}^{B}$. 
The choice of the constant $\Tilde{B}$ such that $\mathcal{F}_{\Omega}^{B}\subseteq \mathcal{H}_{\Omega}^{\Tilde{B}}$  is satisfied depends on the encoding strategy. 
As written in \cite{Caro2021encodingdependent}, it can for example for integer valued frequencies be chosen as $\Tilde{B}=2B$. 
Thus, Lemma \ref{lemma:Generalizationboundalsovalidforsobolev} can be applied and the following bound holds:
    \begin{align}
        D_{H^k}\left(f^*,f_{h^k}\right) \leq D_{h^k}\left(f^*,f_{h^k}\right)+r(|\Omega|,B,\Tilde{B}, c,I,\delta)\ ,
    \end{align}
    with a function $r(|\Omega|,B,\Tilde{B}, c,I,\delta)$ which tends to $0$ as $I\to\infty$.    
\end{proof}

\begin{customthm}{6}[Generalization bound for $L^p$]
Let $f^*\in \mathcal{F}\subseteq H^k([0,2\pi]^N)$ be a target function, and let there be a $B>0$ and a $\Tilde{B}>0$, such that $\mathcal{F}_{\Omega}^{B}\subseteq \mathcal{H}_{\Omega}^{\Tilde{B}}$ is a suitable model family.
Let us further assume that $\ell_{h^k}(f_1(\bm{x}),f_2(\bm{x}))\leq c$ for all $\mathbf{x}\in [0,2\pi]^N$, and for all $f_1, f_2 \in\mathcal{F}_{\Omega}^{B}$ or $\mathcal{F}$. 
Assume that $k,p\in\mathbb{N}$ satisfy one of the two following cases:
\begin{enumerate}
    \item $N\left(\frac{1}{2}-\frac{1}{p}\right)<k<N/2$ and $1\leq p<N$.
    \item $k\geq N/2$  and $1\leq p<\infty$.
\end{enumerate}
For any $\delta\in(0,1)$ and the empirical risk $D_{h^k}\left(f^*,f\right)$ trained on an i.i.d. training data $S$ with size $I$ and containing data of $\xi$ partial derivatives, the following holds for all functions $f\in \mathcal{F}_{\Omega}^{B}$ with probability at least $1-\delta$:
\begin{align}
    \dfrac{1}{C}D_{L^p}\left(f^*,f\right) \leq D_{h^k}\left(f^*,f\right)+r(|\Omega|,\xi,B,\Tilde{B}, c,I,\delta),
\end{align}
where $C$ is a constant and $r(|\Omega|,\xi,B,\Tilde{B}, c,I,\delta)\to 0$ as $I\to\infty$.
\end{customthm}

\begin{proof}
We will prove the theorem by proving the following two inequalities:
\begin{align}
    \dfrac{1}{C}D_{L^p}\left(f^*,f\right) \leq D_{H^k}\left(f^*,f\right) \leq D_{h^k}\left(f^*,f\right)+r(|\Omega|,B,\Tilde{B}, c,I,\delta)\ .
\end{align}
The right hand side inequality is following directly from Theorem \ref{thm:sobolev_loss_Hk}, and the left hand side inequality is a consequence of Theorem \ref{thm:Rellich-Kondrachov}.
Let us look at case $1$ in Theorem \ref{thm:Rellich-Kondrachov}:
We want to rewrite the bound $p< 2N/(N-2k)$ as an upper bound for $k$ for a given $p$. Let us therefore firstly check, which values $p$ is allowed to reach. Due to $k$ being bound from above by $k<N/2$, the upper bound on $p$, $p< 2N/(N-2k)$ 
is maximal for $k=\frac{N}{2}-1$, in which case the upper bound on $p$ becomes $p<N$. That means that values for $p$ chosen in $1\leq p<N$ are valid values. With $p$ such chosen, the bound
$p< 2N/(N-2k)$ is equivalent to bounding $k$ in the following way: 
\begin{align}
N\left(\frac{1}{2}-\frac{1}{p}\right)<k\ .
\end{align}
For case 2 in Theorem \ref{thm:Rellich-Kondrachov}, we have the inequalities
 $k\geq N/2$  and $1\leq p<\infty$.\\
Further, because of the assumptions $\ell_{h^k}\left(f^*(\mathbf{x}),f(\mathbf{x})\right)\leq c$ for all $\mathbf{x}\in [0,2\pi]^N$, the subdomain $U=[0,2\pi]^N$ is equal to $U_0$, and because an $N$-dimensional plane in $\mathbb{R}^N$ is $\mathbb{R}^N$ itself, $U$ is also equal to $U_0^N$. 
Let $C$ be a cone of height at most $\pi$, angle at most $\pi/2$. Then, for each $\bm{x}$ in $U=[0,2\pi]^N$, we can choose an appropriate axis direction such that $C_x$ lies entirely in $U$, so it satisfies the cone condition.
\\To sum up, Theorem \ref{thm:Rellich-Kondrachov} states that for the cases
 \begin{enumerate}
    \item $N\left(\frac{1}{2}-\frac{1}{p}\right)<k<N/2$ and $1\leq p<N$.
    \item $k\geq N/2$  and $1\leq p<\infty$,
\end{enumerate}
the following embeddings are compact:
\begin{align}
    H^{k}([0,2\pi]^N)\to L^{p}([0,2\pi]^N)\ .
\end{align}
According to the definition of a compact embedding (Definition \ref{def:compact_embedding}), 
there exists a constant $C$, such that \begin{align}
    \|f^*-f\|_{L^p}\leq C\|f^*-f\|_{H^k}\ .
    \end{align}
\end{proof}

\begin{customthm}{7}[Generalization bound for $C^0$]
Let $f^*\in \mathcal{F}\subseteq H^k([0,2\pi]^N)$ be a target function, and let there be a $B>0$ and a $\Tilde{B}>0$, such that $\mathcal{F}_{\Omega}^{B}\subseteq \mathcal{H}_{\Omega}^{\Tilde{B}}$ is a suitable model family.
Let us further assume that $ \ell_{h^k}(f_1(\bm{x}),f_2(\bm{x}))\leq c$ for all $\mathbf{x}\in [0,2\pi]^N$, and for all $f_1, f_2 \in\mathcal{F}_{\Omega}^{B}$ or $\mathcal{F}$ and that $\|f\|_{\infty}\leq B$ for all $f \in\mathcal{F}_{\Omega}^{B}$.
Assume, that $k\in\mathbb{N}$ satisfies $k>N/2$.
For any $\delta\in(0,1)$ and the empirical risk $D_{h^k}\left(f^*,f\right)$ trained on an i.i.d. training data $S$ with size $I$ and containing data of $\xi$ partial derivatives, the following holds for all functions $f\in \mathcal{F}_{\Omega}^{B}$ with probability at least $1-\delta$:
\begin{align}
   \frac{1}{C}D_{C^0}\left(f^*,f\right) \leq D_{h^k}\left(f^*,f\right)+r(|\Omega|,\xi,B,\Tilde{B}, c,I,\delta),
\end{align}
where $C$ is a constant and $r(|\Omega|,\xi,B,\Tilde{B}, c,I,\delta)\to 0$ as $I\to\infty$.
\end{customthm}

\begin{proof}
The prove of this theorem is equivalent to the proof of Theorem \ref{thm:sobolev_loss_Lp} above. We will prove this theorem as well by proving the following two inequalities:
\begin{align}
    \dfrac{1}{C}D_{L^p}\left(f^*,f\right) \leq D_{H^k}\left(f^*,f\right) \leq D_{h^k}\left(f^*,f\right)+r(|\mathcal{M}|,I,\delta)\ .
\end{align}
The right hand side inequality is following directly from Theorem \ref{thm:sobolev_loss_Hk}, and the left hand side inequality is a consequence of Theorem \ref{thm:Rellich-Kondrachov}. As written in the proof of Theorem \ref{thm:sobolev_loss_Lp}, the assumptions of Theorem \ref{thm:Rellich-Kondrachov} are satisfied, we can thus also apply it here.\\
The upper bound on the distance $D_{C^0}\left(f^*,f\right)$ in the supremum norm is a direct consequence of the third case in Theorem \ref{thm:Rellich-Kondrachov}.

\end{proof}

\end{document}